\documentclass[aps,prl,superscriptaddress,reprint]{revtex4-1}
\usepackage[utf8]{inputenc}

\usepackage{dsfont}
\usepackage{microtype}
\usepackage{amsthm}
\usepackage{amssymb}
\usepackage{mathtools}
\usepackage{physics}
\usepackage{subfigure}
\usepackage{xspace}
\usepackage[colorlinks]{hyperref}
\usepackage{cleveref}

\usepackage{graphicx}
\usepackage{tikz}
\usetikzlibrary{patterns, decorations.pathreplacing}

\definecolor{colour}{RGB}{225,225,225}
\definecolor{colour1}{RGB}{255,255,255}
\definecolor{colour3}{RGB}{1,105,201}
\definecolor{colour2}{RGB}{255,184,9}
\definecolor{colour4}{RGB}{253,51,1}
\definecolor{colour0}{RGB}{0,0,0}
\definecolor{colour$\ast$}{RGB}{255,255,255}

\input{figures/tiles}

\newcommand{\Tbw}[4]{\tile{\edgeN{white}{#1}\edgeE{white}{#2}\edgeS{white}{#3}\edgeW{white}{#4}}}
\def\Tw(#1,#2,#3,#4) {\Tbw{#1}{#2}{#3}{#4}}

\newcommand\Tomg[4]{\tile{\edgeN{colour#1}{#1}\edgeE{colour#2}{#2}\edgeS{colour#3}{#3}\edgeW{colour#4}{#4}}}
\def\T(#1,#2,#3,#4) {\Tomg{#1}{#2}{#3}{#4}}
\def\TT(#1,#2,#3,#4) {\begin{tiles}*\Tomg{#1}{#2}{#3}{#4}\end{tiles}}

  \def\Snes(#1,#2,#3) {
    \starN{colour#1}{#1}
    \starE{colour#2}{#2}
    \starS{colour#3}{#3}
  }
  \def\SSew(#1,#2) {\begin{tiles}*
      \starE{colour#1}{#1}
      \starW{colour#2}{#2}
  \end{tiles}}
  \def\SSnew(#1,#2,#3) {\begin{tiles}*
      \starN{colour#1}{#1}
      \starE{colour#2}{#2}
      \starW{colour#3}{#3}
  \end{tiles}}
  \def\SSne(#1,#2) {\begin{tiles}*
      \starN{colour#1}{#1}
      \starE{colour#2}{#2}
  \end{tiles}}
\def\SSnes(#1,#2,#3) {\begin{tiles}* \Snes(#1,#2,#3) \end{tiles}}
  \def\SSesw(#1,#2,#3) {\begin{tiles}*
      \starE{colour#1}{#1}
      \starS{colour#2}{#2}
      \starW{colour#3}{#3}
  \end{tiles}}
  \def\SSns(#1,#2) {\begin{tiles}*
      \starN{colour#1}{#1}
      \starS{colour#2}{#2}
  \end{tiles}}
  \def\SSnsw(#1,#2,#3) {\begin{tiles}*
      \starN{colour#1}{#1}
      \starS{colour#2}{#2}
      \starW{colour#3}{#3}
  \end{tiles}}
  \def\SS(#1,#2,#3,#4) {\begin{tiles}*
      \starN{colour#1}{#1}
      \starE{colour#2}{#2}
      \starS{colour#3}{#3}
      \starW{colour#4}{#4}
  \end{tiles}}
  \def\S(#1,#2,#3,#4) {
    \starN{colour#1}{#1}
    \starE{colour#2}{#2}
    \starS{colour#3}{#3}
    \starW{colour#4}{#4}
  }

  \newcommand\tuple[2]{\ensuremath{#1#2}}

  \newcommand{\mt}[1]{\quad\text{#1}\quad}

  \newcommand{\field}[1]{\ensuremath{\mathds{#1}}}
  \newcommand{\identity}{\ensuremath{\mathds{1}}}
  \renewcommand{\op}[1]{\ensuremath{\mathbf{#1}}}
  \newcommand{\bounded}{\ensuremath{\mathcal B}}
  \newcommand{\hs}{\ensuremath{\mathcal H}}

  \newcommand{\lmin}{\ensuremath{\lambda_\text{min}}}
  \newcommand{\tc}{\textsc{tc}}
  \newcommand{\cl}{\textsc{cl}}

  \newcommand{\mcl}{\mathcal}

  \DeclareMathOperator{\lcm}{lcm}

  \newtheorem{theorem}{Theorem}
  \newtheorem{lemma}[theorem]{Lemma}
  \newtheorem{corollary}[theorem]{Corollary}

  \makeatletter
  \newcommand*{\balancecolsandclearpage}{%
    \close@column@grid
    \clearpage
    \twocolumngrid
  }
  \makeatother

\begin{document}
  \title{Size-Driven Quantum Phase Transitions}

  \author{Johannes Bausch}
  \affiliation{DAMTP, University of Cambridge, Centre for Mathematical Sciences, Cambridge CB3 0WA, United Kingdom}

  \author{Toby S. Cubitt}
  \affiliation{Department of Computer Science, University College London, Gower Street, London WC1E 6BT, United Kingdom}

  \author{Angelo Lucia}
  \affiliation{Departamento de An\'alisis Matem\'atico, Universidad Complutense de Madrid, 28040 Madrid, Spain}
  \affiliation{QMATH, Department of Mathematical Sciences, University of Copenhagen, Universitetsparken 5, 2100 Copenhagen, Denmark}
  \affiliation{NBIA, Niels Bohr Institute, University of Copenhagen, Blegdamsvej 17, 2100 Copenhagen, Denmark}

  \author{David \surname{Perez-Garcia}}
  \affiliation{Departamento de An\'alisis Matem\'atico, Universidad Complutense de Madrid, 28040 Madrid, Spain}
  \affiliation{IMI, Universidad Complutense de Madrid, 28040 Madrid, Spain}
  \affiliation{ICMAT, Calle Nicol\'as Cabrera, Campus de Cantoblanco, 28049 Madrid}

  \author{Michael M. Wolf}
  \affiliation{Department of Mathematics, Technische Universit\"at M\"unchen, 85748 Garching, Germany}

  \begin{abstract}
    Can the properties of the thermodynamic limit of a many-body quantum system be extrapolated by
    analysing a sequence of finite-size cases? We present models for which such an approach gives
    completely misleading results: translationally invariant, local Hamiltonians on a square
    lattice with open boundary conditions and constant spectral gap, which have a classical product
    ground state for all system sizes smaller than a particular threshold size, but a ground state
    with topological degeneracy for all system sizes larger than this threshold.  Starting from a
    minimal case with spins of dimension 6 and threshold lattice size $15\times15$, we show that the
    latter grows faster than any computable function with increasing local spin dimension. The
    resulting effect may be viewed as a new type of quantum phase transition that is driven by the
    size of the system rather than by an external field or coupling strength. We prove that the
    construction is thermally robust, showing that these effects are in principle accessible to
    experimental observation.
  \end{abstract}

  \maketitle



  The thermodynamic limit of many-body quantum Hamiltonians is the predominant mathematical tool used to study macroscopic properties of physical systems. In order to understand the properties of a physical model, it is important to distinguish and recognise features that are a consequence of \emph{finite-size effects}, i.e.\ properties of the model which are not present in the thermodynamic limit but appear as a by-product of conditions which only hold for systems sizes smaller than some threshold. While some finite-size effects only produce small perturbations of the real model, this is not always the case. For example, relevant finite size effects for the distinct behaviour of antiferromagnets on even or odd system sizes have been proposed in~\cite{Lounis} and recently observed experimentally in~\cite{Guidi}.

  In this work we show that finite-size effects can in fact be dominant at arbitrary length scales, to the point of completely obscuring the physics of the thermodynamic limit. This phenomenon occurs not just in pathological examples, but even e.g.\ for translationally invariant Hamiltonians on low-dimensional spins arranged on a square lattice.

  {\it {\bf Main result 1:} We explicitly construct models exhibiting the following exotic finite-size effects: below a threshold lattice size with sides of length $N$, the ground state of the Hamiltonian is a non-degenerate product state in the canonical basis, i.e.~entirely classical, with a constant spectral gap above it. For system sizes greater than $N$, however, the low energy space is that of the Toric Code, which is in a sense as quantum as possible: the ground state exhibits topological degeneracy, and the system has anyonic excitations.

  Moreover, we calculate threshold lattice sizes $N$ for models with local spin dimensions up to $10$ (see \cref{tab:main-result}). Already for dimension $10$ the threshold size can be as large as $5.2\cdot10^{36534}$.}

\

  Since in practice in a real-world experiment the ground state cannot be accessed, and only the Gibbs state at some small but non-zero temperature can be prepared, we also prove for one family of models that:

  {\it {\bf Main result 2:} For any given measurement precision, there exists
    a finite temperature below which measurements on system sizes smaller than
    the threshold can not distinguish the thermal state from a classical state
    (i.e.\ product in the canonical basis), while the thermodynamic limit
    converges for low temperatures to the ground state of the Toric
    Code~\cite{1608.04449v1}.
    Even for measurements with errors of magnitude no larger than $10^{-6}$,
    the required temperatures are rather mild (see \cref{tab:main-result}).}

  \

  This sudden and dramatic change in the nature of the ground state may be viewed as a type of \emph{quantum phase transition}, driven by the system size rather than a varying external field or coupling strength.

  It has been known for some time that the critical values of external parameters (e.g.\ temperature, pressure) can depend on the size of the studied samples.
  Well-studied effects include rising melting points for small particles~\cite{Buffat1976,Goldstein1992}, structural temperature- or pressure-dependent phase transitions between different crystal lattices in thin-film samples and in nano-crystals~\cite{Tolbert1994,McHale1997,Rivest2011,Li2016}, where the energetically favourable structure differs from that in the thermodynamic limit.
  And charge density wave order transitions or superconductivity~\cite{Xi2015,Yu2015}, for which the critical temperature changes when approaching mono-layer sample sizes.

  Here, we exhibit a transition which is driven by the system size itself; the transition occurs at some critical size, without any external parameters varying at all.
  The effects which are most reminiscent of what we prove rigorously here are certain peculiar phenomena for mono-layer samples, or samples with 3 or 13 atom layers, for which the described phase transition cannot be observed anymore~\cite{Xi2015,Yu2015}; one suggested explanation is a lack of space for nucleation sites~\cite{Tolbert1994,Li2016}.

  \begin{table}
    \begin{ruledtabular}
      \begin{tabular}{r l l l l l l}
        $d$ & 4     & 6     & 7     & 8    & 9             & 10                   \\
        $N_d$ & 2     & 15    & 84    & 420  & $3.3\cdot10^7$ & $5.2\cdot10^{36534}$ \\
        $T_d [\frac{\Delta}{k_B}]$ & 0.058 & 0.050 & 0.043 & 0.038 & 0.020           & 5.9$\mu$
      \end{tabular}
    \end{ruledtabular}
    \caption{Threshold lattice sizes $N_d\times N_d$ for different spin dimension $d$ in our constructions. For lattice sizes larger than this threshold, finite-size effects
      suddenly disappear and the physics of the thermodynamic limit becomes accessible. Up to dimension $8$, a prime periodic Wang tiling construction gives a larger threshold size than embedding of Busy Beaver Turing machines. The critical temperature $T_d$ gives an estimate for the temperature at which the transition can still be discriminated with a fidelity of $1-10^{-6}$, as a function of the system's spectral gap $\Delta$, which here is equal to the interaction strength since the Hamiltonians are commuting.
      \label{tab:main-result}}

\end{table}

  \Cref{tab:main-result} shows an overview of the explicit examples we construct. The threshold
  system sizes $N_d$ from these examples show that large thresholds are possible with relatively
  small $d$-dimensional spins. These are of course lower bounds on the maximum possible threshold
  size for given local spin dimension; even larger size thresholds may be achievable by other
  constructions.
  Though our constructions can be straightforwardly generalised to produce size driven-transitions
  to other quantum phases, we have chosen to focus on a transition from classical to Toric Code for
  several reasons. First, because this makes our constructions stable against temperature and
  extensive perturbations. Second, and even more important, because the topological order present in
  the Toric Code allows to claim rigorously the existence of a phase transition even at finite
  sizes.

  In order to prove these effects mathematically rigorously, we deliberately construct examples for which there exists an analytic solution. However, this is not true for the general case: as the structure of the Hamiltonian becomes more complex, one expects the behaviour to become more erratic. Indeed, we know that for extremely complex Hamiltonians with very large local spin dimension the behaviour can even become uncomputable~\cite{Spectralgapundecidability}.

  It is important to emphasise that the dramatic finite-size effects exhibited here do \emph{not} depend on any careful tuning of coupling strengths, and occur for Hamiltonians without obvious separation of energy scales in their coupling constants or the matrix entries of the local interactions. Without this restriction, i.e.\ allowing interactions of magnitude $\order{1}$ and $\order{1/N^2}$, it is in fact trivial to construct a model whose ground state changes character at system size $\order{N}$, with the spectral gap closing as $\order{1/N}$. Our result is much stronger, in the sense that it does not allow such a prediction based solely on an analysis of the coupling strengths, nor from extrapolation of spectral data; in particular, the spectral gap of our model remains constant all the way up to the transition.

  \subsection{Hamiltonian Construction}
  For local spin dimension $d>3$, we construct a local, translationally invariant spin Hamiltonian $\op H^{(d)}$ on a 2D square lattice with open boundary conditions, such that there exists a threshold system size $N_d\times N_d$, up to which the ground state of $\op H^{(d)}$ is entirely classical (i.e.\ product in the canonical basis), whereas for larger lattice sizes the ground state is that of the Toric Code.
  The \emph{transition thresholds} $N_d$ for given local dimension $d$ in our explicit constructions are shown in \cref{tab:main-result}. For $d>6$, we give a general procedure for constructing models for which $N_d$ grows faster than any computable function.

  The \emph{Toric Code}---introduced by Kitaev~\cite{Kitaev}---is defined by a Hamiltonian on a two-dimensional spin-$1/2$ lattice. It is one of the simplest models exhibiting topological order~\cite{Wen2013,pachos2012introduction}.

  \begin{figure}
    \begin{center}
      \begin{tikzpicture}[
          draw=black,
          star/.style={
            draw=colour3,
            line width=1pt
          },
          plaquette/.style={
            draw=colour4,
            line width=1pt
          }
        ]
        \foreach \x in {0,1,2,3,4,5} {
          \foreach \y in {1,2,3} {
            \path[fill=black] (\x+.5,\y) circle[radius=.05] (\x,\y+.5) circle[radius=.05];
          }
          \path[fill=black] (\x+.5,4) circle[radius=.05];
        }
        \foreach \y in {1,2,3} {
          \path[fill=black] (6,\y+.5) circle[radius=.05];
        }
        \foreach \x in {0,1,2,3,4,5,6} {
          \draw (\x,1) -- (\x,4);
        }
        \foreach \y in {1,2,3,4} {
          \draw (0,\y) -- (6,\y);
        }

        \draw[plaquette] (1,2) rectangle (2,3);
        \path[plaquette,fill=colour4] (1.5,2) circle[radius=.05] (1,2.5) circle[radius=.05] (2,2.5) circle[radius=.05] (1.5,3) circle[radius=.05];
        \node at (1.51,2.47) {$p$};

        \draw[star] (3.5,2) -- (4.5,2);
        \draw[star] (4,1.5) -- (4,2.5);
        \path[star,fill=colour3] (3.5,2) circle[radius=.05] (4.5,2) circle[radius=.05] (4,1.5) circle[radius=.05] (4,2.5) circle[radius=.05];
        \node at (4.2,2.2) {$s$};

      \end{tikzpicture}
    \end{center}
    \caption{\leavevmode
      Plaquette and star interactions of the two-dimensional Toric Code Hamiltonian $\op H_\tc$. We assign a spin-$1/2$ particle to every lattice edge marked with a dot. $\op H_\tc$ is a sum of $4$-local interactions, plaquettes and stars, which are products of $\sigma^z$ and $\sigma^x$ operators, respectively.
    \label{fig:toriccode}}
  \end{figure}
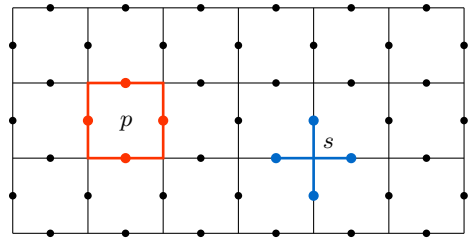

  We start out with a finite lattice as shown in \cref{fig:toriccode}. To every edge marked with a dot, we assign a $d$-dimensional spin $\field C^d=\hs_\tc\oplus\hs_\cl$ where $\hs_\tc=\field C^2$ and $\hs_\cl=\field C^{d-2}$, such that the overall Hilbert space on the lattice is a tensor product over all separate spins, i.e.
  \[\hs^{(d)}=\bigotimes(\hs_\tc\oplus\hs_\cl) \cong (\bigotimes \hs_\tc) \oplus (\bigotimes\hs_\cl) \oplus \hs',\]
  where $\hs'$ contains all mixed $\hs_\tc$ and $\hs_\cl$ terms.

  We define a purely classical Hamiltonian $\op H_\cl$ with support only on the subspace $\bigotimes\hs_\cl$, such that the ground state energy of $\op H_\cl$ is $-1$ for lattice sizes $N\le N_d$, and otherwise $\lmin(\op H_\cl)\ge1/2$. We then combine $\op H_\cl$ with $\op H_\tc$ in such a way that the spectrum below some energy $\lambda'>0$ is uniquely determined by one or other of these Hamiltonians, by giving an energy penalty for any state with support on $\hs'$. We define the overall Hamiltonian by $\op H^{(d)}:=\op H_\tc + \op H_\cl + \op H'$, where
  \[
    \op H':=C\sum_{i \sim j} \identity_\tc^i \otimes \identity_\cl^j + \identity_\cl^i \otimes \identity_\tc^j,
  \]
  where $i\sim j$ stands for a sum over any adjacent spins. $\identity_\tc$ denotes the projector on the $\hs_\tc$ subspace, and analogously $\identity_\cl$. Note that $\op H'$ only contains $2$-local interactions.

  In this way, any state $\ket\psi\in\hs^{(d)}$ supported on $\hs'$ will necessarily pick up an energy penalty of at least $C$. Choosing $C=1+\lmin(\op H_\cl)$ shifts this part of the spectrum to energies $\ge 1$.
  We can rescale $\op H_\tc$ to have its low-energy spectrum within $[0,1/2]$. The ground state of $\op H^{(d)}$ will thus be given by either $\op H_\tc$ or $\op H_\cl$, whichever has the smaller energy. In particular, the system will change abruptly from classical to topologically ordered with anyonic excitations when the lattice size $N$ surpasses the threshold $N_d$, while keeping a constant spectral gap.

  In order to construct a suitable classical Hamiltonian $\op H_\cl$, we will exploit the same
  locality structure as in the Toric Code--$4$-local star and plaquette interactions---since this
  does not increase the interaction range of the overall Hamiltonian $\op H^{(d)}$.  We will present
  two different constructions, based on a generalised tiling problem, which will give rise to
  different scaling of the threshold lattice size: one will be called \emph{Periodic Tiling}, the
  other \emph{Turing Machine Tiling}.
  We will only consider the case of open boundary conditions, which is the most natural one in this context.

  \paragraph{Tiling Hamiltonians.}
  It is convenient to express the interactions as a so-called \emph{tiling} problem with extra constraints, similar to the well-known Wang tiles. A Wang tile is simply a square tile with coloured edges, and the condition for placing two tiles next to each other is that their edge colours match. Despite this simple setup, it has been shown that the question of whether one can tile the entire plane with a finite set of Wang tiles is in fact undecidable~\cite{berger1966undecidability}, which shows that tiling can encode extremely complex behaviour.

  \begin{figure*}
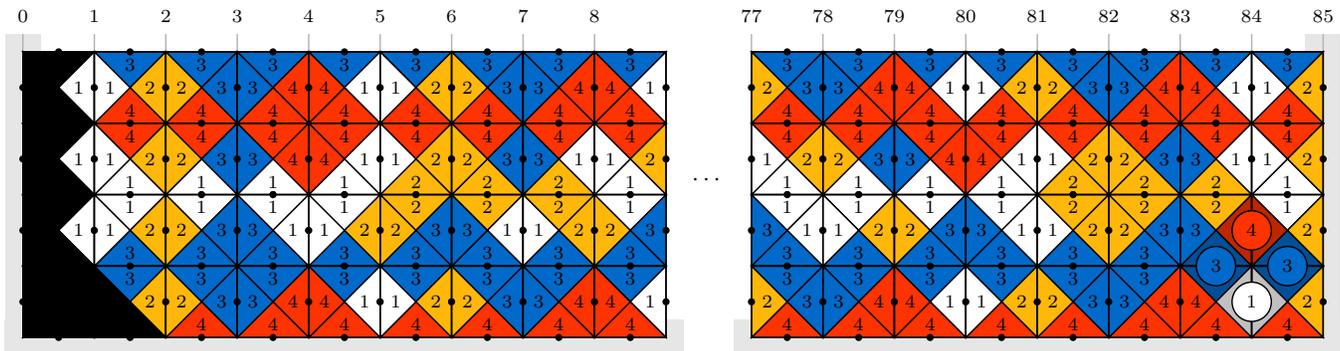

    \begin{tiles}*
      \path[fill=black!10] (.75,-3.25) rectangle (1.25,1.25);
      \path[fill=black!10] ((.75,-3.25) rectangle (10.25,-2.75);
      \foreach \x in {0,...,8}
      {
        \draw[black!30] (\x+1,-2.25) -- (\x+1,1.25);
        \node[] at (\x+1,1.5) {$\x$};
      }

      \T(0,1,0,0)
      \T(3,2,4,1)
      \T(3,3,4,2)
      \T(3,4,4,3)
      \T(3,1,4,4)
      \T(3,2,4,1)
      \T(3,3,4,2)
      \T(3,4,4,3)
      \T(3,1,4,4)
      \$

      \T(0,1,0,0)
      \T(4,2,1,1)
      \T(4,3,1,2)
      \T(4,4,1,3)
      \T(4,1,1,4)
      \T(4,2,2,1)
      \T(4,3,2,2)
      \T(4,1,2,3)
      \T(4,2,1,1)
      \$

      \T(0,1,0,0)
      \T(1,2,3,1)
      \T(1,3,3,2)
      \T(1,1,3,3)
      \T(1,2,3,1)
      \T(2,3,3,2)
      \T(2,1,3,3)
      \T(2,2,3,1)
      \T(1,3,3,2)
      \$

      \T(0,0,0,0)
      \T(3,2,0,0)
      \T(3,3,4,2)
      \T(3,4,4,3)
      \T(3,1,4,4)
      \T(3,2,4,1)
      \T(3,3,4,2)
      \T(3,4,4,3)
      \T(3,1,4,4)

    \end{tiles}
    $\cdots$
    \begin{tiles}*
      \path[fill=black!10] (8.75,-3.25) rectangle (9.25,1.25);
      \path[fill=black!10] (.75,-3.25) rectangle (9.25,-2.75);
      \foreach \x in {77,...,85}
      {
        \draw[black!30] (\x-76,-2.25) -- (\x-76,1.25);
        \node[] at (\x-76,1.5) {$\x$};
      }

      \T(3,3,4,2)
      \T(3,4,4,3)
      \T(3,1,4,4)
      \T(3,2,4,1)
      \T(3,3,4,2)
      \T(3,4,4,3)
      \T(3,1,4,4)
      \T(3,2,4,1)
      \$

      \T(4,2,1,1)
      \T(4,3,1,2)
      \T(4,4,1,3)
      \T(4,1,1,4)
      \T(4,2,2,1)
      \T(4,3,2,2)
      \T(4,1,2,3)
      \T(4,2,1,1)
      \$

      \T(1,1,3,3)
      \T(1,2,3,1)
      \T(1,3,3,2)
      \T(1,1,3,3)
      \T(2,2,3,1)
      \T(2,3,3,2)
      \T(2,4,3,3)
      \T(1,2,3,4)
      \$

      \T(3,3,4,2)
      \T(3,4,4,3)
      \T(3,1,4,4)
      \T(3,2,4,1)
      \T(3,3,4,2)
      \T(3,4,4,3)
      \T(3,1,4,4)
      \T(3,2,4,1)

      \path[fill=black, opacity=.25] (7,-2) -- (8,-1) -- (9,-2) -- (8,-3) -- cycle;
      \begin{scope}[
          shift={(7.5,-2.5)}
        ]
        \S(4,3,1,3)
      \end{scope}
    \end{tiles}
    \caption{\leavevmode
      Section of the prime periodic pattern for four plus one colours, which is $84$-periodic. The left edge and lower corner is enforced by giving the solid black square in the corner a bonus of 1, but penalising black to appear to the right of black on a horizontal edge: this way, the global pattern can be started uniquely with open boundary conditions. The horizontal edge colours form disjoint sets: starting from the bottom row, the colours are red $\{4\}$, blue $\{3\}$, white and yellow $\{1,2\}$, after which the cycle continues with red. This can be achieved using the 4-local star interactions, e.g.\ by allowing blue to only appear next to blue and white next to white. For the top row allowing two colours white and yellow, we alternate between them whenever the colour on the vertical edge above it is white. Within each row, these colours on the vertical edge count cyclically through sub-sequences of length $4$, $3$ and $4+3=7$, respectively, which yields the overall horizontal period $\lcm\{4,3,7\}=84$.\\
      Every $84$ tiles, the pattern necessarily exposes a unique local colour configuration, highlighted in the lower right corner. It can be penalised by a single star interaction and forces the spectrum of the associated Hamiltonian to $\ge1$ when the system size surpasses the threshold $N_5=84$.
    \label{fig:5colourpattern}}
  \end{figure*}

  It is easy to represent the tiling problem as a ground state energy problem of a classical,
  translationally invariant Hamiltonian $\op H_{\mathcal W}$ on the lattice in \cref{fig:toriccode}, and straightforward to verify that this representation only defines a single energy scale. As shown in \cref{fig:4colour}, each tile can be regarded as a plaquette on the lattice. The condition that neighbouring tiles share the same edge colour is thus automatically met. It is clear that for $c$ colours, we need a $c$-dimensional classical subspace $\hs_\cl$ for each spin, i.e.\ $d=c+2$. Working on this classical subspace, we want to find local Hamiltonian plaquette interactions between the spins surrounding a plaquette $p$---which we denote with $E_p$---that penalise any tile not in our set of allowed tiles $\mathcal W$.
  To achieve this, we define a local classical tile interaction via
  \begin{equation}
    \op h^{(p)}:=\sum_{w\in \mathcal W}a_w\bigotimes_{e\in E_p}\ketbra{w_{e,p}},
  \end{equation}
  where $w_{e,p}$ labels the colour on edge $e$ of tile $w$ placed at plaquette site $p$. The parameters $(a_w)_{\mathcal W}$ do not depend on the plaquette position, and the overall translationally-invariant tiling Hamiltonian is given by the sum over all plaquette sites in the lattice
  $
  \op H_\mathcal W:=\sum_{p}(\identity -\op h^{(p)})
  $.
  If $a_w=1$ for all $w$, one can show that $\op H_\mathcal W$ has ground state energy zero if and only if the set $\mathcal W$ tiles the lattice. If we want to give an energy ``bonus'' to (i.e.\ decrease the energy of) a specific tile $w$, we can set $a_w>1$. An energy penalty can be given by setting $a_w<1$. Each tiling thus has a net score---bonuses minus penalties minus mismatching tile pairs. The net score of a specific tiling gives the energy of the corresponding state of $\op H_\mathcal W$. In general, then, the ground state of $\op H_\mathcal W$ will maximise the number of tiles with a bonus while avoiding as many penalties as possible.

  A similar construction allows us to add extra star-shaped interactions, constraining tile edges adjacent to a corner. The overall Hamiltonian $\op H_\mathcal W+\op H_\mathcal S$ will then have an optimal ground state in the sense that the sum of penalties minus the sum of bonuses---for both tiles and stars---is minimised. The rigorous argument is presented in \cref{lemma:tiling-embedding} in the appendix.

  \paragraph{Periodic Tiling.}
  With as few colours $c$ as possible, we create a set of tiles and stars which permit a unique periodic tiling pattern in the ground state. The construction for any number of colours $c\ge3$ is described in the appendix, and an example for five colours can be seen in \cref{fig:5colourpattern}. One can show that this period grows at least exponentially with the number of colours \footnote{One can show that $N_{q+3}=p_{q}\#=\Omega(\exp((1+o(1))q\ln q))$, where $\#$ denotes the primorial function and $p_q$ the $q$\textsuperscript{th} prime.}. By penalising a pattern that occurs precisely once per global period---highlighted in \cref{fig:5colourpattern}---we can ensure that the ground state spectrum jumps to $\lmin(\op H_\cl)\ge 1$ for any larger square size. The transition threshold $N_d=N_{c+2}$ for this model is thus given simply by the horizontal pattern period.

  \paragraph{Turing Machine Tiling.}
  Starting from a number of colours $c\ge 6$, it becomes possible to embed a Turing machine into a set of tile and star interactions. We improve on an idea introduced by Robinson~\cite{robinson1971undecidability}---which has been exploited in~\cite{Spectralgapundecidability} to show undecidability of the spectral gap---by making use of the extra star constraints to significantly reduce the necessary local dimension.
  In this new construction, the transition threshold $N_d$ grows faster than any computable function and surpasses the threshold from the periodic tiling construction for $c\ge7$.

  A \emph{Turing machine} is an abstract machine for algorithmic computation proposed by Turing~\cite{Turing1937}, and is generally accepted as the standard mathematical model for formalising problems of computability and complexity. Such a machine is defined by a finite set of instructions. It is equipped with a finite internal memory, together with a two-way infinite tape where it can read and write symbols, and it can move left or right on the tape. The instructions tell the machine how to update the symbol at the current tape location, and which way to move on the tape, depending on the symbol it reads from the tape and its current internal state. In the field of computational complexity, the hardness of a problem is usually defined in terms of the amount of resources needed for a Turing machine to solve it, in terms of time needed and tape consumed.

  A Turing machine halts if its internal memory reaches a specified ``halting'' state, after which no further updates take place. We say a given machine is \emph{halting} if it eventually reaches a halting state. If we restrict to machines with a fixed number of states $q$ for its internal memory, and which read and write only two symbols, i.e.\ $0$ and $1$, then the set of possible halting machine programs is finite: there has to exist one that runs for longer, or at least as long as, any other. These machines are called \emph{Busy Beavers}, and their running time is called the \emph{Busy Beaver number} $S(q)$. It is known that $S(q)$ grows faster than any computable function~\cite{rado1962non} \footnote{\label{busy-beaver}A related quantity, which is also called the Busy Beaver number but is usually denoted $\Sigma(q)$, is defined as the largest number of non-blank symbols written out by the machine before terminating, and is a lower bound on $S(q)$. $\Sigma(q)$ also grows faster than any computable function.}.

  As in the case of the periodic tiling, we find a way of embedding a Busy Beaver Turing machine into the ground state of a classical Hamiltonian: as soon as the Busy Beaver halts, there will be a penalty, since at that point there is no valid way to continue updating the tape. The tiling is thus possible up to a square size of at least $S(q)/\sqrt 2$ \footnote{The constant factor of $\sqrt 2$ is due to the fact that in our construction the tape of the Turing machine is encoded diagonally with respect to the square lattice, and the head of the machine follows a zig-zag pattern, which in the worst case scenario can only reach a distance of $S(q)/\sqrt 2$ from the origin.}. As we need $c=q+2$ colours for a $q$ state Busy Beaver, we immediately find a transition threshold of $N_{q+4}\ge S(q)/\sqrt 2$.

  \subsection{Stability}

  Our model inherits its stability against extensive perturbations from the Toric Code: as we show
  in the appendix, in the thermodynamic limit the ground states of $\op H^{(d)} + \op V$ are in the
  same phase as the ground states of our model, for a sufficiently weak local perturbation $\op
  V$. Before the threshold lattice size, by standard perturbation theory, there exists a maximal
  perturbation strength (which will depend on the critical size $N_d$), below which the spectral gap
  of the perturbed Hamiltonian will not close: the perturbed ground state can be continuously
  deformed to the unperturbed one with local unitaries \cite{Bachmann_2011}, and thus it will still
  have the properties of a classical state.

  Moreover, in the appendix we also show that there exists a finite temperature $T_d$ below which the
  thermal state of the Hamiltonian will still be very close to the ground state for any system size
  up to the threshold $N_d$, meaning that any measurement will still reveal a classical state up to
  very small errors.  The temperature $T_d$ depends only inverse-logarithmically on the threshold
  size, and therefore its scaling with $d$ will be mild in the case of the prime periodic
  tiling. \Cref{tab:main-result} lists the temperatures corresponding to the various local
  dimensions $d$, which is a linear function of $\Delta/k_B$, where $\Delta$ is the spectral gap of
  the Hamiltonian and $k_B$ the Boltzmann constant. Since our models are commuting Hamiltonians, the
  spectral gap $\Delta$ is simply equal to the strength of the interactions between the spins.

  For the prime periodic tiling, we also show that if we go to the thermodynamic limit at finite
  temperature, and then send the temperature to zero (a procedure which is a more mathematically
  correct description of real implementations~\cite{Aizenman1981}), we will recover only ground
  states of the Toric Code. This shows that there is a complete disagreement between the
  mathematical predictions from the thermodynamic limit, and any measurement performed on systems
  below the threshold size. The Busy Beaver construction can be modified in order to show the same
  property, using the original construction of Robinson~\cite{Robinson1971}, at the cost of greatly
  increasing the local dimension.

  \subsection{Conclusion}
  By constructing two concrete classes of examples, we have shown that there exist translationally invariant, local Hamiltonians on a 2D square lattice with constant spectral gap and open boundary conditions, which belong to a topologically ordered phase in the thermodynamic limit, but appear to be classical for finite system sizes smaller than a certain threshold. This sudden change in the nature of the ground state---with constant spectral gap right up to the threshold size---is reminiscent of a phase transition, but one driven not by any external parameters, but by the size of the system. Furthermore, we have proven mathematically rigorously that these results are thermally robust: this ``size-driven'' transition can be observed at non-zero temperature, and the temperatures required to observe this even at very high precision are rather mild.

  The threshold size can grow extremely fast as a function of the local spin dimension---for one class it grows faster than any computable function---showing that even for physically reasonable systems with low local dimension, size-dependent behaviour can sometimes occur at system sizes that are inaccessible experimentally or numerically.

  A common approach to understanding the physical properties of many-body models in the thermodynamic limit is to analyse a growing sequence of finite system sizes---numerically or experimentally---and then extrapolate the characteristics of interest to the macroscopic limit~\cite{samaj2013introduction}. This approach has proven highly successful in numerous cases~\cite{LandauLifshitzVolume5, march1992electron,domb1983phase,Pirvu2012,Tagliacozzo2008}.
  Numerical simulations of lattice models play a key role in understanding the dynamics
  of a system, e.g.\ in lattice gauge theories~\cite{KogutLatticeQCD}, fluid
  dynamics~\cite{QuantumLatticeGas} and condensed matter physics~\cite{landau2001computer}. All
  these simulations are computationally intensive, so accessible lattice sizes are
  severely limited---e.g.\ for heavy quark simulations, current lattices have sizes reaching
  $96^3\times192$~\cite[Ch. 18]{Agashe:2014kda} (the larger dimension representing time).
  On the other hand, it has been shown that e.g.\ determining whether a system is gapped or gapless in the thermodynamic limit is an undecidable problem~\cite{Spectralgapundecidability}, albeit for extremely contrived and artificial models with extremely large local spin dimension.

  Our results show that there exist classes of local, physical systems on a 2D lattice of spins with moderate dimension for which---without some specialised analysis that takes into account the type of phenomena we have presented---it is impossible to tell with certainty whether the system behaves the same on macroscopic scales as it does for finite sizes. Whilst it is unlikely such pathological behaviour occurs in the types of system listed above, our results show that there exist new and interesting physical phenomenon that are not amendable to this kind of scaling analysis.

  Many variations of these results are possible. It is easy to e.g.\ reverse our construction and transition from topologically ordered at low system sizes to classical for large lattices, and it is clear that similar constructions using a different Hamiltonian than the Toric Code are possible. As usual when exotic models are found, we expect that the ability to switch properties of a Hamiltonian on and off depending on the system size could also lead to interesting applications in future.

  \begin{acknowledgments}
    J.\,B.\ acknowledges support from the German National Academic Foundation and the EPSRC (grant 1600123). T.\,S.\,C.\ is supported by the Royal Society.
    A.\,L.\ and D.\,P.\,G.\ acknowledge support from  MINECO (grant MTM2014-54240-P) and Comunidad de Madrid (grant QUITEMAD+-CM, ref. S2013/ICE-2801).
    A.\,L.\ acknowledges support from MINECO fellowship FPI BES-2012-052404, the European Research Council (ERC Grant Agreement no 337603), the Danish Council for Independent Research (Sapere Aude) and VILLUM FONDEN via the QMATH Centre of Excellence (Grant No. 10059).
    D.\,P.\,G.\ and M.\,M.\,W.\  acknowledge support the European CHIST-ERA project CQC (funded partially by MINECO grant PRI-PIMCHI-2011-1071).
    This work was made possible through the support of grant \#48322 from the John Templeton Foundation. The opinions expressed in this publication are those of the authors and do not necessarily reflect the views of the John Templeton Foundation.
    This project has received funding from the European Research Council (ERC) under the European Union's Horizon 2020 research and innovation program (grant agreement No 648913).
  \end{acknowledgments}

  \bibliographystyle{plain}
  \bibliography{bibliography}

  \balancecolsandclearpage
  \subsection{Embedding a Generalised Tiling into a Hamiltonian Spectrum}
  In this section, we rigorously formulate the embedding of the tiling problems we consider in this work into the spectrum of a local Hamiltonian. Instead of focusing only on star and plaquette interactions, we take an abstract point of view and define the notion of a \emph{generalised tiling}. Assume $\mathcal G=(V,E)$ is a finite undirected graph with coloured vertices, where we allow colours $C:=\{1,\ldots,c\}$, $c\in\field N$. Let $\mathcal L:=\{l:l\subset\mathcal G\}$ be a finite set of (local) interactions, e.g.\ all the $3$- or $4$-local star and plaquette interactions on a lattice as in \cref{fig:toriccode}. For all interactions $l\in\mathcal L$, we allow a finite set of \emph{pieces} $\mathcal T_l:=\{(c_v)_{v\in l}\}$---where the family $(c_v)_{v\in l}$ assigns a colour to every vertex in $l$---and a weight function $w_l:\mathcal T_l\rightarrow\field R$. Now assign a colour to each vertex in $\mathcal G$, e.g.\ by defining a family $(c_v)_{v\in V}$, $c_v\in C$. The \emph{score} of this assignment is then given by
  \[
    \mathrm{score}:=\sum_{l\in\mathcal L}\begin{cases}
      1-w_l(t_l) &\text{if $(c_v)_{v\in l}$ is a valid piece in $\mathcal T_l$}\\
      1 &\text{otherwise}.
    \end{cases}
  \]
  For $w_l(t_l)<1$, we can thus give a score penalty, and $w_l(t_l)>1$ gives a bonus to piece $t_l$ at site $l$. An assignment $w_l(t_l)=1$ is neutral and gives neither bonus nor penalty. Observe that \emph{not} including a piece in the piece set $\mathcal T_l$ is equivalent to giving it a weight of $0$. It is easy to see how this specialises to our tiling examples: in case of the periodic tiling and for $l$ a plaquette interaction in the bulk, the sets $\mathcal T_l$ would all be identical and correspond to the allowed $4$-local tiles. The $w_l$ then assign the bonuses or penalties, accordingly.

  We formulate the following lemma.
  \begin{lemma}\label{lemma:tiling-embedding}
    Define a Hilbert space $\hs:=\bigotimes_{v\in V}\field C^c$ over the interaction graph $\mathcal G$, assigning $c$-dimensional qudits to each vertex $v\in V$. Then there exists a classical Hamiltonian $\op H$ on $\hs$, diagonal in the computational basis, with $\mathcal L$-local interactions such that the eigenvalue $\lambda$ for a basis state $\ket\psi=\bigotimes_{v\in V}\ket{c_v}$ is given by the score of the associated generalised tiling, i.e.
    \[
      \lambda=\sum_{l\in\mathcal L}\begin{cases}
        1-w_l(t_l) &\text{if $\ket\psi|_l\in\mathcal T_l$}\\
        1 &\text{otherwise}.
      \end{cases}
    \]
    We denote with $\ket\psi|_l$ the restriction of $\ket\psi$ to the subspace $\bigotimes_{v\in l}\field C^c\le\hs$.
  \end{lemma}
  \begin{proof}
    Define
    \[
      \op H:=\sum_{l\in\mathcal L}\left(\identity-\sum_{t\in\mathcal T_l}w_l(t)\Pi_t\right),
    \]
    where $\Pi_t:=\bigotimes_{v\in l}\ketbra{t_v}$ denotes the projector onto the valid piece $t\in\mathcal T_l$ for interaction $l\in\mathcal L$, and $t_v$ denotes the colour of vertex $v$ for piece $t$. Take a computational basis state $\ket\psi=\bigotimes_{v\in V}\ket{c_v}$. Then
    \begin{align*}
      \op H\ket\psi&=\sum_{l\in\mathcal L}\left(\ket\psi-\sum_{t\in\mathcal T_l}\ket\psi\begin{cases}
      w_l(t_l) &\text{if $\ket\psi|_l\in\mathcal T_l$}\\
      0 &\text{otherwise}
  \end{cases}\right)\\
  &=\lambda\ket\psi,
\end{align*}
and the claim follows.
\end{proof}
This allows us to conclude the following corollary.
\begin{corollary}
  The ground state energy of $\op H$ is determined by the lowest score assignment of the associated generalised tiling problem.
\end{corollary}
Equipped with this machinery, it suffices to formulate generalised tiling problems on the square lattices as in \cref{fig:toriccode} with $3$- and $4$-local interactions, such that for lattice sizes below some threshold, the lowest score assignment has a score $\le-1/2$, and above the threshold the lowest score assignment has a score $\ge 1$. This way, combining the Toric Code Hamiltonian $\op H_\tc$ via \cref{lemma:combining-hamiltonian} creates a model with a size-induced transition from classical to topological ground state. Observe that we require our model to be translationally invariant.

\subsection{The Toric Code}
The Toric Code Hamiltonian $\op H_\tc$ is a sum of $3$- and $4$-local interactions
\[
  \op H_\tc:=-J\sum_s \op A^{\!(s)}-J\sum_p\op B^{(p)},
\]
with
$\op A^{\!(s)}:=\prod_{i\in s}\sigma_i^x$ a product of Pauli $\sigma^x$ acting on $4$ spins $i$ adjacent to vertex $s$ as seen in \cref{fig:toriccode}. The $\op B^{(p)}:=\prod_{i\in p}\sigma_i^z$ are defined analogously. We call the $\op A^{\!(s)}$ \emph{star} and the $\op B^{(p)}$ \emph{plaquette}-interactions, respectively. The free parameter $J>0$ is a coupling strength and can be used to rescale the spectrum.

\subsection{Prime Period Tiling}
The key idea is to create a tiling pattern that can tile the entire plane with a very large period $p$. We require that a certain locally detectable sub-pattern---i.e.\ using a star interaction---occurs exactly once per period. By disallowing this sub-pattern, the tiling will be possible up to a square of size $p\times p$, but once the grid surpasses this size, there will be at least one pattern violation, which can be penalised locally with a Hamiltonian term.

\paragraph{General Construction.}
For the general construction, we first regard the following discrete optimisation problem. Assume we have $q$ colours available.
We want to construct a family of tuples $(r_i)_{1\le i\le f}$, each of which stands for a row of colours $r_i=(r_{i1},r_{i2},\ldots,r_{im_i})$. These rows have to satisfy three constraints.
\begin{enumerate}
  \item There are fewer than $q$ rows overall, i.e.\ $f\le q$.
  \item Each row has fewer than $q$ colours, i.e.\ $m_i\le q\ \forall i$.
  \item For the first and last row, each colour $r_{ij}$ is picked from the $q$ colours available, i.e.\ $r_{1j},r_{fj}\in\{1,\ldots,q\}$---for all other rows, we leave out the last, i.e.\ $r_{ij}\in\{1,\ldots,q-1\}$.
\end{enumerate}
We can associate a period $p_i:=\sum_j^{m_i}r_{ij}$ to each row $i$. The rows $r_i$ are now chosen such that the objective function $p(q):=\lcm\{p_i:i\le I\}$---i.e.\ the overall period---is maximised.

We now give a description on how to translate such an optimal row family $(r_i)_i$ into a set of tiles and stars that enforces a unique horizontal tile sequence with periodicity $p(q)$. More specifically, for each row, we define tiles that allow a colour pattern
\begin{equation}\label{eq:pattern}
  1,\ldots,r_{i1},1,\ldots,r_{i2},\ldots,1,\ldots,r_{i(m_i-1)},1,\ldots,r_{im_i}
\end{equation}
on their vertical edge---i.e.\ the tiles form sub-periods, starting at 1 and counting to $r_{ij}$, then starting at 1 again, counting to $r_{i(j+1)}$ etc.. Cleverly choosing colours on the horizontal edges then a) make this pattern unique for each row, and b) enforce a unique stacking order of the rows overall---which in turn yields a $p(q)$-periodic global tile pattern.

To facilitate the tile notation, we use a few shorthands.
\begin{itemize}
  \item The last sub-period for each row has highest colour $r_{im_i}=:\bar r_i$.
  \item We sequentially enumerate the sub-periods with colours for use on the horizontal edges, i.e.\ $r_{11}\leftrightarrow 1, r_{12}\leftrightarrow 2,\ldots,\bar r_1\leftrightarrow m_1,r_{21}\leftrightarrow m_1+1$ etc.. The highest such label for every row is denoted with $h_i$, and the lowest $l_i$, e.g.\ for the first row $l_1=1$ and $h_1=m_1$. More rigorously, we have the sequences $h_0=0$, $h_i=h_{i-1}+m_i$ and $l_i=h_{i-1}+1$.
  \item The set of colours on the horizontal edges needed for the $i$\textsuperscript{th} row is denoted with $V_i:=\{l_i,\ldots,h_i\}$, respectively.
\end{itemize}

For every row $r_i$, we then define the tiles
\[
  \begin{tiles}*
    \Tw(t,c',b,c)
  \end{tiles}
  \quad\begin{minipage}{6.5cm}\raggedright
    $\forall t \in V_{(i-1\bmod f)},\quad b \in V_{(i \bmod f)}$ \\
    $\forall c=1,\ldots,r_{ij}$,
  \end{minipage}
\]
where $c'=c+1\mod{r_{ij}}$. As an example, consider the first row with $V_1=\{1,\dots, m_1\}$. We obtain a set of tiles
\[
  \begin{tiles}*
    \Tw($\ast_f$,2,$1$,1)
    \Tw($\ast_f$,3,$1$,2)
  \end{tiles}
  \cdots
  \begin{tiles}*
    \Tw($\ast_f$,1,$1$,$r_{11}$)
    \Tw($\ast_f$,2,2,1)
  \end{tiles}
  \cdots
  \begin{tiles}*
    \Tw($\ast_f$,1,2,$r_{12}$)
  \end{tiles}
  \cdots
  \begin{tiles}*
    \Tw($\ast_f$,1,$n_1$,$\bar r_1$)
  \end{tiles},
\]
where $\ast_f$ stands for any colour allowed on the bottom of the last row, i.e.\ $\ast_f\in V_f$. All other rows are defined analogously, where the top and bottom colours are chosen successively, i.e.\ for the $i$\textsuperscript{th} row, we use colours $V_i$ for the bottom and any $V_{i-1}$ for the top.

On their own, the tiles from different sets can be mixed at will. To enforce that each row can only be assembled from its own tile set, for each row $i$, we restrict to the following star configurations:
\begin{align*}
  &\begin{tiles}*
  \starN{white}{$c$}
  \starE{white}{$j$}
  \starW{white}{$j$}
  \starS{white}{$\ast$}
\end{tiles}
\quad\ 2\le c\le r_{ij}, \quad j \in V_i, \\
&\begin{tiles}*
\starN{white}{$1$}
\starE{white}{$j'$}
\starW{white}{$j$}
\starS{white}{$\ast$}
\end{tiles}
\mt{and}
\begin{tiles}*
  \starN{white}{$1$}
  \starE{white}{$j$}
  \starW{white}{$j'$}
  \starS{white}{$\ast$}
\end{tiles}
\quad\forall\ j \in V_i
\end{align*}
where $j':=l_i+(j+1\mod m_i)$. It is easy to verify that each row defines a unique $p_i$-periodic horizontal tile pattern as in \cref{eq:pattern}.
As the top colours of the $i$\textsuperscript{th} row are restricted to the bottom colours of the $i-1$\textsuperscript{th} row---modulo the numbers of rows $f$---the rows can be stacked above each other uniquely. Every block of rows $r_1,\ldots,r_f$, stacked vertically, thus defines a valid horizontally $p$-periodic tiling pattern for the plane.

In order to be able to detect this periodicity locally, we make use of the extra colour available in all but the first and last row due to constraint (3). For all $i=2,\ldots,f-1$, we add two tiles
\[
  \begin{tiles}*
    \Tw($h_{i\!-\!1}$,$d$,$h_i$,$\bar r_i$)
  \end{tiles}
  \mt{and}
  \begin{tiles}*
    \Tw($l_{i\!-\!1}$,2,$l_i$,$d$)
  \end{tiles}
\]
Alternative to the row sequence $\ldots,\bar r_i-1,\bar r_i,\mathbf 1,2,\ldots$, this allows counting $\ldots,\bar r_i-1,\bar r_i,\mathbf q,2,\ldots$. By adding the star penalties
\[
  \begin{tiles}*
    \starN{white}{1}
    \starS{white}{1}
    \starE{white}{1}
    \starW{white}{$h_1$}
  \end{tiles}
  \mt{and}
  \begin{tiles}*
    \starN{white}{q}
    \starS{white}{1}
    \starE{white}{$l_i$}
    \starW{white}{$h_i$}
  \end{tiles}
  \quad\forall\ 2<i<f-1,
\]
we ensure that whenever two consecutive rows complete a cycle in the same column, we mark the occurrence with a $q$ instead of a $1$. This way, if in the first row we finish a cycle with a $1$ and observe a $q$ right below, we know that the entire horizontal pattern has completed one period. To be more specific, every $p(q)=\lcm\{p_1,\ldots,p_K,p_L,p_f\}$ tiles, where $L=f-1$ and $K=f-2$, we have the pattern
\[
  \begin{tiles}*
    \Tw($h_K$,$q$,$h_L$,$\bar c_L$)
    \Tw($l_K$,$2$,$l_L$,$q$)
    \\ \$
    \Tw($h_L$,$1$,$h_f$,$\bar c_f$)
    \Tw($l_L$,2,$l_f$,$1$)
  \end{tiles},
  \mt{locally detectable via}
  \begin{tiles}*
    \starN{white}{$q$}
    \starS{white}{1}
    \starW{white}{$h_L$}
    \starE{white}{$l_L$}
  \end{tiles}.
\]
We call this sub-pattern a \emph{period marker} and penalise it with a weight of $2$.

So far we have constructed a tile set which can periodically tile the entire plane. By disallowing said period marker, we restrict the tileable square size to at most $p\times p$. Observe however that due to the freedom to shift sets of rows horizontally and the entire pattern vertically, there are a potentially huge number of possibilities to tile any square smaller than $p\times p$. We will thus add a special colour to fix this freedom, borrowing an idea from~\cite{gottesman2009quantum}. We will enforce a specific pattern in the lower left corner, which uniquely fixes the starting configuration for the bulk, without imposing any boundary condition, but instead by adding bulk interactions which will have the effect of favouring the desired configuration in the boundary. We add the following tiles:

\[
  \begin{tiles}*
    \T(0,0,0,0)
  \end{tiles}
  ,\quad
  \begin{tiles}*
    \edgeN{black}{}
    \edgeE{white}{$1$}
    \edgeS{black}{}
    \edgeW{black}{}
  \end{tiles}
  \mt{and}
  \begin{tiles}*
    \edgeN{white}{$3$}
    \edgeE{white}{$2$}
    \edgeS{black}{}
    \edgeW{black}{}
  \end{tiles}.
\]
We further disallow black appearing to the right of black using a star constraint, and give a bonus of $1/2$ to the all-black tile. It is then easy to verify that the best score tiling starts with the following configuration in the lower left corner:
\[
  \begin{tiles}
    \Tw(,,,,)
    \begin{scope}[
        shift={(1,0)}
      ]
      \edgeN{black}{}
      \edgeE{white}{$1$}
      \edgeS{black}{}
      \edgeW{black}{}
    \end{scope}

    \$

    \Tw(,,,,)
    \begin{scope}[
        shift={(1,-1)}
      ]
      \edgeN{black}{}
      \edgeE{white}{$1$}
      \edgeS{black}{}
      \edgeW{black}{}
    \end{scope}

    \$

    \T(0,0,0,0)
    \Tw(,,,,)
    \begin{scope}[
        shift={(2,-2)}
      ]
      \edgeN{white}{$3$}
      \edgeE{white}{$2$}
      \edgeS{black}{}
      \edgeW{black}{}
    \end{scope}

  \end{tiles}
\]

It is straightforward to verify that starting from this corner configuration, the plane can be tiled uniquely up to a grid size of $p\times p$ with a net score of $-1/2$, after which the net penalty jumps to a value $\ge 1$.

In \cref{tab:periods}, we list the cases $q=2,\ldots,8$ with a solution to the associated constraint problem and the resulting overall period $p$.  It is easy to see that in the case of 2 colours only, the extra black tile to remove degeneracies is redundant.

\begin{table}
  \begin{ruledtabular}
    \begin{tabular}{lll}
    	colours $q+1$ & line periods     & overall period $p$ \\ \hline
    	2$^\dagger$   & 2, 2             & 4                  \\
    	4             & 3, 5             & 15                 \\
    	5             & 4, 3, 7          & 84                 \\
    	6             & 5, 4, 3, 7       & 420                \\
    	7             & 6, 5, 7, 11      & 2310               \\
    	8             & 7, 6, 5, 11, 13  & 30030              \\
    	9             & 8, 7, 11, 13, 15 & 120120
    \end{tabular}
  \end{ruledtabular}
  \caption{Maximum tiling periods for a given number of colours (plus one special black colour needed for $q>2$). The second column shows how to distribute the periods between the lines. $^\dagger$For 2 colours, it is easy to see that the extra black tile is redundant.
  \label{tab:periods}}
\end{table}

\paragraph{Five Colour Tiling Example.}
We give the $q=4$ colour case as an explicit example. The tile set in \cref{fig:4colour} defines three disjoint sets of tiles, each of which can be assembled into horizontal lines, which in turn can be stacked above each other in a unique order. To avoid mixing the tiles from different sets on one same line, we add the following star operators with parameter $b_s =1$
\begin{figure}
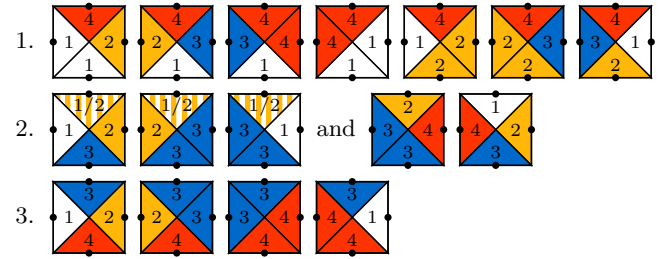

  \begin{tabular}{ll}
    1.&
    \TT(4,2,1,1)
    \TT(4,3,1,2)
    \TT(4,4,1,3)
    \TT(4,1,1,4)
    \TT(4,2,2,1)
    \TT(4,3,2,2)
    \TT(4,1,2,3)
    \\
    2.&
    \begin{tiles}*
      \tile{
        \edgeN{colour1}[colour2]{$1/2$}
        \edgeE{colour2}2
        \edgeS{colour3}3
        \edgeW{colour1}1
      }
    \end{tiles}%
    \begin{tiles}*
      \tile{
        \edgeN{colour1}[colour2]{$1/2$}
        \edgeE{colour3}3
        \edgeS{colour3}3
        \edgeW{colour2}2
      }
    \end{tiles}%
    \begin{tiles}*
      \tile{
        \edgeN{colour1}[colour2]{$1/2$}
        \edgeE{colour1}1
        \edgeS{colour3}3
        \edgeW{colour3}3
      }
    \end{tiles}
    and
    \TT(2,4,3,3)
    \TT(1,2,3,4)
    \\
    3.&
    \TT(3,2,4,1)
    \TT(3,3,4,2)
    \TT(3,4,4,3)
    \TT(3,1,4,4)
  \end{tabular}
  \caption{Four$+1$ colour tile set that defines a tiling of the plane with period $84$. The striped top in the second row denotes either of the colours $1$ or $2$.}
  \label{fig:4colour}
\end{figure}

\begin{align*}
  &\SS(2,1,$\ast$,1) ,\
  \SS(3,1,$\ast$,1) ,\
  \SS(4,1,$\ast$,1) ,\ \\
  &\SS(2,2,$\ast$,2) ,\
  \SS(3,2,$\ast$,2)
  \mt{and}
  \SS(1,1,$\ast$,2) ,\
  \SS(1,2,$\ast$,1) , \\
  &\SS($\ast$,3,$\ast$,3) ,\
  \SS($\ast$,4,$\ast$,4) .
\end{align*}
The third row then unambiguously assembles to a line with a horizontal period of $4$, while the vertical edges appearing on the first line periodically cycle through $1,2,3,4,1,2,3$ with a period of $7$. The horizontal edges of the second line are fixed by the line above and below, but there exists some freedom to choose the colours on the vertical edge. More specifically, we use the freedom of either counting $\ldots,2,3,\mathbf 1,2,\ldots$ or $\ldots,2,3,\mathbf 4,2,\ldots$ to detect when all three lines complete a period in the same column. We add a penalty for the configuration
\[
  \SS(1,1,1,2)
\]
by adding the corresponding operator multiplied by $b_s=-1$, which enforces colour $4$ to appear instead of colour $1$ whenever the first row finishes one cycle at the same time as the second one below. The combined period $p$ of the three lines is thus given by $\lcm(7,4,3)=84$, and it can be detected by penalising the configuration
\begin{equation}\label{eq:periodmarker}
  \SS(4,3,1,3)
\end{equation}
with a penalty of $2$ (i.e., with $b_s = -2$.)

The freedom to horizontally shift the lines relative to each other or the entire pattern vertically is fixed without adding boundary conditions with the special tiles
\[
  \begin{tiles}*
    \T(0,0,0,0)
  \end{tiles}
  ,\quad
  \begin{tiles}*
    \T(0,1,0,0)
  \end{tiles}
  \mt{and}
  \begin{tiles}*
    \T(3,2,0,0)
  \end{tiles}.
\]
We choose $a_w = 2$ for the first. Starting from there, the entire plane can be tiled uniquely up to a grid size of $84\times84$, after which the penalised star---\cref{eq:periodmarker}---naturally occurs and the net penalty is $\ge1$.
A section of the complete $5$-colour tiling can be seen in \cref{fig:5colourpattern}.

Generalising the prime tiling to higher dimensions, we obtain the periods as given in \cref{tab:periods}.

\subsection{Turing Machine Tiling}
\begin{table}
  \begin{ruledtabular}
    \begin{tabular}{llll}
    	states $|Q|$ & colours $c$ & $S(|Q|)$                               & threshold $N_{d}$            \\ \hline
    	3            & 6           & 21~\cite{lin1965computer}              & 14                           \\
    	4            & 6           & 107~\cite{TM42}                        & 75                           \\
    	5            & 7           & $\ge 4.7 \cdot 10^7$~\cite{TM52}       & $3.3 \cdot 10^7$             \\
    	6            & 8           & $\ge 7.4 \cdot 10^{36534}$~\cite{TM62} & $5.2 \cdot 10^{36534}$
    \end{tabular}
  \end{ruledtabular}
  \caption{Number of $2$ symbol Turing machine states $Q$ and tile colours $c=\max\{|A|^2+2,|A|+|Q|\}$ required for the embedding. $S(|Q|)$ is the corresponding Busy Beaver number, and $N_{d}$ denotes the lattice threshold size, where $d=c+2$.}
  \label{tab:busy-beaver}
\end{table}

A Turing machine is given by finite sets of states $Q$ and symbols $A$ with a transition function
$ \delta : Q \times A \to Q \times A \times \{ \text{left}, \text{right} \}$,
representing the set of instructions of the Turing machine. The machine is equipped with a tape, which is sequence of symbols arranged in a 1-dimensional line extending indefinitely in both directions, and initialised with a special ``blank'' symbol (which we will denote $0$ for simplicity of notation). The machine has an internal state $q \in Q$ and a head which sits over one of the symbols of the tape: at each step, the head reads the symbols $s$ underneath, and it will write the symbol $\bar s$, change its internal state to $\bar q$ and then move in direction $d \in \{ \text{left}, \text{right} \}$, where $(\bar q,\bar s, d) = \delta(q,s)$.
The machine starts in an initial state $q_0 \in Q$ and \emph{halts} if there is no forward transition for a given tuple $(q,s)$ \footnote{Such definition is equivalent to defining a special halting state $h$ for which no further transition is defined}.

We will show know how to construct a set of plaquette and star interactions in such a way that the ground state encodes the history of a run of a given Turing machine. The construction will involve the use of Wang tilings and of star interactions: the latter will allow us to greatly reduce the spin dimension needed by previous works which where based only on the tiling problem~\cite{robinson1971undecidability,berger1966undecidability}.

We assume the lattice is oriented as in \cref{fig:tmevolution}, where the Turing machine starts in the lower left corner.
We chose to encode the tape of the Turing machine at a given time step in diagonal direction across the square lattice (denoted by the thick grey lines in \cref{fig:tmevolution}), and movement in the orthogonal direction represents the time evolution. Moreover, we store the head position and internal state of the Turing machine by including the state $q$ in the tape, on the right of the symbol which will be read by the machine. Using this convention, we have that---even if the tape space is finite---it is extended by one symbol in both directions at each time step, and therefore the tape available will always be sufficient for the machine to run.

We will interpret colours on horizontal and vertical edges differently---horizontal either as pair of symbols $s_1s_2\in A\times A$ or special boundary colour
$\begin{tiles}* \path[draw=black,fill=colour2] (0,0) circle[radius=.20] node {$\tuple00$};\end{tiles}$
  or
$\begin{tiles}* \path[draw=black,fill=colour3] (0,0) circle[radius=.20] node {$\tuple00$};\end{tiles}$
  , and vertical either as symbol $s\in A$ or state $q\in Q$, the latter of which we highlight in red.

  As we did in the periodic tiling construction, we use special bulk interaction (which will be present everywhere in the lattice) to constrain the left and bottom boundaries. In order to do so, we use the following interactions
  \[
    \begin{tiles}*
      \edgeN{colour2}{\tuple00}
      \edgeS{colour2}{\tuple00}
      \edgeW{white}0
      \edgeE{white}{$\ast$}
    \end{tiles}
    \mt{,}
    \begin{tiles}*
      \starW{colour3}{\tuple00}
      \starE{colour3}{\tuple00}
    \end{tiles}
    \mt{and}
    \begin{tiles}*
      \edgeN{colour2}{\tuple00}
      \edgeE{colour4}{$q_0$}
      \edgeS{colour3}{\tuple00}
      \edgeW{white}{$0$}
    \end{tiles}.
  \]
  Note that the 2-local blue term is indeed a bulk interaction, and not a ``cut-off'' boundary term.
  By giving the last of these tiles a single bonus of 2---similar to the all black tile for the periodic tiling---we obtain precisely one choice for the left and lower grid edges, namely all $0$s as in \cref{fig:tmevolution}; in particular, the last shown tile with initial state $q_0$ correctly initialises the Turing machine in the lower left corner. This valid initial configuration defines the unique highest-net-bonus tiling possible.

  To avoid cases where we validly tile the plane without a TM head---i.e.\ with net bonus zero---we use a star interaction to give a bonus of $1/2$ for any \emph{white} symbol on a vertical edge
appearing to the left or right of another arbitrary symbol, i.e.
$$
1\times\begin{tiles}*
      \starW{white}{\tuple**}
      \starE{white}{\tuple**}
\end{tiles}\mt{,}
\frac12\times\begin{tiles}*
      \starW{white}{\tuple**}
      \starE{colour2}{\tuple00}
\end{tiles}
\mt{and}
\frac12\times\begin{tiles}*
      \starW{colour2}{\tuple00}
      \starE{white}{\tuple**}
\end{tiles}.
$$
We further give a penalty of $1$ for the white symbol appearing anywhere.  This way, in the bulk, the net contribution of $2\times 1/2-1=0$ for each of the white edges, whereas if they appeared on the left end of the plane a net penalty $\ge1/2$ would be inflicted.  A similar combination of bonus and penalty terms allows us to ensure that the lower edge is blue, and all other configurations obtain a net penalty of $\ge 1/2$ as well. Like that, there exists no configuration with net penalty $<1/2$ without the initial $q_0$ tile in the lower left corner.
  From now on, we treat the boundary symbols
$\begin{tiles}* \path[draw=black,fill=colour2] (0,0) circle[radius=.20] node {$\tuple00$};\end{tiles}$
  and
$\begin{tiles}* \path[draw=black,fill=colour3] (0,0) circle[radius=.20] node {$\tuple00$};\end{tiles}$
  as equivalent to
$\begin{tiles}* \path[draw=black,fill=white] (0,0) circle[radius=.20] node {$\tuple00$};\end{tiles}$.

  To implement the transitions rules we effectively need 6 different spins to interact (three for each time step): in fact, if the tape around the head reads $s, q, t$ for some $q \in Q$ and some $s,t \in A$, then it has to be updated to $\bar q, \bar s, t$ if $\delta(q,s) = (\bar q,\bar s, \text{left})$, while it has to be updated to $\bar s, t, q$ if instead $\delta(q,s) = (\bar q,\bar s, \text{right})$.

  Since we only have at our disposition 4-body interactions, in order to implement an effective 6 body interaction we will make use of the extra register we have allowed for in the horizontal edges, which will allow to ``synchronise'' a plaquette and a star interaction, as shown in \cref{fig:tmevolution}. This is done by defining,
  for every transition, a pair of tiles and stars, i.e.\ if $\delta(q,s)=(\bar q,\bar s,\text{left})$
  \begin{equation}
    \begin{tiles}*
      \starN{colour4}{$\bar q$}
      \starS{colour4}{$q$}
      \starW{white}{$\tuple s\ast$}
      \starE{white}{$\tuple{\bar s}s$}
    \end{tiles}
    \qand
    \begin{tiles}*
      \tile{
        \edgeN{white}{$\tuple{\bar s}s$}
        \edgeS{white}{$\tuple w\ast$}
        \edgeW{colour4}{$q$}
        \edgeE{white}{$w$}
      }
    \end{tiles},
  \end{equation}
  where $\ast$ represents any symbol, and for an analogous right transition
  \begin{equation}
    \begin{tiles}*
      \starN{white}{$\bar s$}
      \starS{colour4}{$q$}
      \starW{white}{$\tuple s*$}
      \starE{white}{$\tuple ws$}
    \end{tiles}
    \qand
    \begin{tiles}*
      \tile{
        \edgeN{white}{$\tuple ws $}
        \edgeS{white}{$\tuple w*$}
        \edgeW{colour4}{$q$}
        \edgeE{colour4}{$\bar q$}
      }
    \end{tiles}.
  \end{equation}
  Observe how the symbol pairs are necessary to uniquely couple the pair of interactions to obtain the left and right transition depicted in \cref{fig:tmevolution}, but are disregarded for any successive transition.
  The rest of the tape which is not affected by the transition rules has to be copied verbatim to the next time step.
  Implementing such bookkeeping tiles is straightforward, as we only need to take care of the extra---and in this situation unused---register in the horizontal edges, which is discarded when copying to vertical edges, and set to the ``blank'' symbol when copying from vertical to horizontal edges. More precisely, for every $a,b \in A$, we define
  \begin{equation}
    \begin{tiles}*
      \starN{white}{$a$}
      \starS{white}{$b$}
      \starW{white}{$\tuple a*$}
      \starE{white}{$\tuple b0$}
    \end{tiles}
    \qand
    \begin{tiles}*
      \tile{
        \edgeN{white}{$\tuple a0 $}
        \edgeS{white}{$\tuple b*$}
        \edgeW{white}{$a$}
        \edgeE{white}{$b$}
      }
    \end{tiles}.
  \end{equation}
  Overall, this construction thus requires $c=\max\{|A|^2+2, |A|+|Q|\}$ colours.
  It is easy to verify that starting from the initial tile, a square can be uniquely tiled with net bonuses $1$ if and only if the Turing machine does \emph{not} halt within its boundaries. All other tilings necessarily violate at least one constraint and thus have a net penalty $\ge1/2$. A sample evolution can be seen in \cref{fig:tmevolution}.

  \begin{figure}
    \centering
    \includegraphics[width=\linewidth]{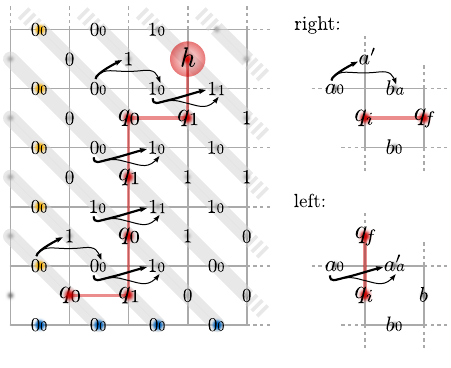}
    \caption{\leavevmode
      Embedding of a Turing machine into a tiling problem with extra star constraints. We chose a representation in which the Turing machine head sits in between the tape symbols, reading and writing the symbol on its left. Every grey horizontal slice shows the tape at one evolution step: it is initialised to $\ldots,0,q_0,0,\ldots$, where $q_0$ is the initial machine state, and every successive step is uniquely defined by the transition rules. Shown here is the 2 state Busy Beaver which halts after 6 steps. Since there is no valid tile with a halting state, the system necessarily frustrates for lattices larger $6/\sqrt 2\approx 4$ tiles on each edge.
      \\
      Each right transition $(q_i,a)\mapsto(q_f,a',\mathrm{right})$ or left transition $(q_i,a)\mapsto(q_f,a',\mathrm{left})$ is translated into a pair of 4-local plaquette and star interactions, as depicted to the right. Observe how in both cases the tile part of the interaction has to know the initial symbol $a$, which is why the star operator creates a temporary copy of it. This copy is shown as small symbols and numbers to the right of the actual tape content and ignored in any following transition.
      \\
      As the available space grows by one symbol in both directions at each step, there is always enough tape available for the Turing machine. The coloured terms are used to initialise this spare tape to $0$. Away from the head, additional interactions are used to copy currently unused tape segments forward. The exact construction with all interaction terms is explained in detail in the appendix.
    \label{fig:tmevolution}}
  \end{figure}

  The maximum number of steps any \emph{halting} Turing machine with $|Q|$ states and $2$ symbols can take before halting is called the \emph{Busy Beaver number} and is denoted by $S(|Q|)$. Defined in~\cite{rado1962non}, it is known to grow faster than any computable function.
  The staggering threshold sizes $N_d$ in \cref{tab:main-result} show that there is no hope to address the question of extrapolating physical properties of a general system solely with an increase in computational power.

  \subsection{Combining Hamiltonian Spectra}
  \begin{lemma}
    \label{lemma:combining-hamiltonian}
    Let $\op H_1$ and $\op H_2$ two local Hamiltonian defined on $\bigotimes_{u \in \Lambda} \field C^{d_1}$ and $\bigotimes_{u \in \Lambda} \field C^{d_2}$ for some interaction graph $\Lambda$. Let further $\mu\in\field R$. Then there exists a Hamiltonian $\op H$ on $\hs=\bigotimes_{u \in \Lambda} \field C^{d_1} \oplus \field C^{d_2}$ with the following properties:
    \begin{enumerate}
      \item Any eigenvector $v$ of $\op H$ with eigenvalue $\lambda\le\mu$ is given by an eigenvector of either $\op H_1$ or $\op H_2$, extended canonically to the larger Hilbert space \hs, with the same eigenvalue $\lambda$.
      \item $\op H$ is translationally invariant if $\op H_1$ and $\op H_2$ are.
      \item $\op H$ contains nearest neighbour interactions and otherwise leaves the interaction range of $\op H_1$ and $\op H_2$ intact.
    \end{enumerate}
  \end{lemma}
  \begin{proof}
    Let $\identity_1$ and $\identity_2$ be the identity operators on $\field C^{d_1}$ and $\field C^{d_2}$, respectively. Let $\delta := 1+\mu$. Define further
    \[ \op H_0 := \delta \sum_{i \sim j} \identity_1^i \otimes \identity_2^j + \identity_2^i \otimes \identity_1^j ,\]
    where $i\sim j$ denotes any neighbouring spin pairs. Set $\op H := \op H_0 + \op H_1' + \op H_2' $, where $\op H_1':=\op H_1\oplus\,\op 0$ and analogously for $\op H_2'$.

    The last two claims are satisfied by construction. To prove the first point, note that $\op H_0$, $\op H_1'$ and $\op H_2'$ commute and thus share a common eigenbasis with spectrum $\sigma(\op H) = \sigma(\op H_0) + \sigma(\op H_1) + \sigma(\op H_2)$.
    Since $\delta > \mu$, any eigenstate of $\op H$ with eigenvalue $\lambda\le\mu$ thus has to be in the kernel $\ker \op H_0\equiv\mathrm{supp}(\op H_1'+\op H_2')=\mathrm{supp}\,\op H_1'\sqcup\mathrm{supp}\,\op H_2' $, and the claim follows.
  \end{proof}

  \section{Thermal stability}
  \subsection{Stability up to Transition Threshold Size}

  We will now show that for both the periodic tiling and the Busy Beaver model there exists a
  finite inverse temperature $\beta_d$ (depending on the local dimension due to its explicit dependence on the threshold size $N_d$),
  above which the thermal state of the Hamiltonian $\rho_\beta =
  \exp(-\beta \op H^{(d)})/Z_\beta$
  will still be very close to a classical state, i.e.\ to the classical ground state of $\op H^{(d)}$.

  We will recall the following observation of Hastings
 ~\cite{Hastings2007}: if $\rho_0$ is the density matrix corresponding
  to the ground state $P_0$ of $\op H$ (i.e.\ $P_0 = (\tr P_0) \rho_0$), then we have the following
  bound:

  \begin{widetext}
    \begin{multline}
      \norm{\rho_{\beta} - \rho_0 }_1 = \tr \abs{ \frac{ e^{-\beta \op H} - Z_\beta \rho_0}{Z_\beta}}
      = \tr \abs{ \frac{e^{-\beta \op H} - e^{-\beta\lambda_0}P_0}{Z_\beta} + \frac{ e^{-\beta\lambda_0}P_0 -
      Z_\beta \rho_0 }{Z_\beta}  } \notag \\
      \le \frac{ \tr \abs{ e^{-\beta \op H} -  e^{-\beta\lambda_0}P_0}}{Z_\beta} + \frac{\tr
      \abs{  e^{-\beta\lambda_0}P_0 - Z_\beta \rho_0}}{Z_\beta}
      = 2 \frac{\abs{Z_\beta -  e^{-\beta\lambda_0} \tr P_0}}{Z_\beta}
      \label{eq:hastings-bound-1}
    \end{multline}
  \end{widetext}
  where $\lambda_0$ is the ground state energy.
  Moreover, since $Z_\beta \ge \tr P_0 e^{-\beta \lambda_0} \ge e^{-\beta \lambda_0}$, we have that
  \begin{equation}
    \label{eq:hastings-bound-2}
    \frac{\abs{Z_\beta - \tr P_0 e^{-\beta \lambda_0}} }{Z_\beta} \le \sum_{\lambda \in \sigma(\op H)\setminus{\lambda_0}} e^{-\beta (\lambda-\lambda_0)},
  \end{equation}
  where eigenvalues are counted with their multiplicity.

  Let us denote by $\Delta$ the spectral gap of $\op H$, and by $\eta(m)$ the
  number of eigenstates with energy in the range $[m \Delta + \lambda_0, (m+1)\Delta + \lambda_0 )$.
    Then, following~\cite{Hastings2007}, if $\op H$ satisfies
    \begin{equation}
      \label{eq:hastings-condition}
      \eta(m) \le \frac{K^m}{m!},
    \end{equation}
    then we can bound the r.h.s. of \cref{eq:hastings-bound-2} by
    \[
      2 \sum_{m=1}^\infty \eta(m) e^{-\beta \Delta m} \le 2\sum_{m=1}^\infty \frac{(K e^{-\beta \Delta})^m}{m!} = 2(e^{K e^{-\beta \Delta}}-1).
    \]

    Since in our case $\op H^{(d)}$ is a commuting Hamiltonian, $\eta(m)$ grows as $\binom{N_d^2}{m} \le {N_d^{2m}}/{m!}$, which implies
    \begin{equation}
      \label{eq:gibbs-state-approx}
      \norm{\rho_\beta - \rho_0}_1 \le 2( e^{N_d^2 e^{-\beta \Delta}} -1 ).
    \end{equation}

    Fix a small $\epsilon > 0$, and let us now choose $\beta_d$ such that
    \[ 2( e^{N_c^2 e^{-\beta_d \Delta}} -1 ) \le \epsilon, \]
    where $N_d$ is the critical system size of $\op H^{(d)}$, meaning that
    \begin{equation}\label{eq:critical-temp}
      \beta_d = \frac{1}{\Delta} [2 \log N_d - \log \log (1+\frac \epsilon 2)].
    \end{equation}
    With this choice of $\beta_d$, we have that for all system sizes $N \le N_d$ and all
    $\beta \ge \beta_d$, the thermal state $\rho_\beta$ is $\epsilon$-close to the ground state of $\op H^{(d)}$,
    which as we have seen is a classical product state.

    On the other hand, if $N>N_d$, from the periodic tiling construction
    we see that the sector of $\op H^{(d)}$ corresponding to the tiling
    Hamiltonian $\op H_\cl$ necessarily picks up an energy penalty every
    period of $N_d$ (it actually picks up even more, given that the
    pattern is repeated vertically with a period corresponding to the
    number of colours, and therefore the energy penalty of every
    $N_d \times N_d$ square is at least $N_d/d$).
    This implies that every eigenstate of $\op H_\cl$ has a strictly positive energy density, and the spectrum of $\op H_\cl$ is contained in $[ (N/N_d)^2, \infty)$. This is not true for the Busy Beaver embedding, as it could be more favourable to terminate the computation, and then simply continue with a blank tape. The energy density decreases to zero in this case. Augmenting the construction with a base layer formed from Robinson tiles, it is however possible to make this Turing machine embedding similarly robust. We refer the reader to~\cite{Robinson1971} and~\cite[ch.~8]{Spectralgapundecidability} for more details.

      \subsection{Thermodynamic Limit}
      Let us now recall that a state in the thermodynamic limit is given by
      a linear, positive and normalised functional $\omega$ on the algebra of
      quasi-local observables $\mathcal A$, which is the (norm closure of the) inductive
      limit of the finite matrix algebras $\bounded_\Lambda = \bounded(\otimes_{i \in \Lambda}
      \hs^{(d)}_i)$, where $\Lambda$ is an ascending sequences of finite
      lattices converging to $\field{Z}^2$~\cite{Bratteli1997,Aizenman1981}.

        \begin{figure}
          \begin{center}
            \begin{tikzpicture}[
                draw=black,
                star/.style={
                  draw=colour3,
                  line width=1pt
                },
                plaquette/.style={
                  draw=colour4,
                  line width=1pt
                }
              ]
              \draw[draw=none,fill=black!30] (2,.8) -- (2,2.5) -- (3.5,4) -- (7.2,4) -- (7.2,3) -- (3,3) -- (3,.8) -- cycle;

              \foreach \x in {0,...,6} {
                \foreach \y in {1,...,4} {
                  \path[fill=black] (\x+.5,\y) circle[radius=.05] (\x,\y+.5) circle[radius=.05];
                }
                \path[fill=black] (\x+.5,5) circle[radius=.05];
              }
              \foreach \y in {1,...,4} {
                \path[fill=black] (7,\y+.5) circle[radius=.05];
              }
              \foreach \x in {0,...,7} {
                \draw (\x,.8) -- (\x,5.2);
              }
              \foreach \y in {1,...,5} {
                \draw (-.2,\y) -- (7.2,\y);
              }

              \draw[very thick] (3,.8) -- (3,3) -- (7.2,3);

              \foreach \x/\y in {2/.5,2/1.5, 3.5/3,4.5/3,5.5/3,6.5/3} {
                  \path[plaquette,fill=colour4] (\x,\y+1) circle[radius=.05];
              }

              \foreach \x/\y in {2.5/1,2.5/2,2.5/3, 3/3.5,4/3.5,5/3.5,6/3.5,7/3.5} {
                  \path[star,fill=colour3] (\x,\y) circle[radius=.05];
              }

            \end{tikzpicture}
          \end{center}
          \caption{\leavevmode
            Boundary $\partial\Lambda$ of a smooth-bounded rectangular region $\Lambda$ of the spin lattice.  Depicted in red are spins that share a plaquette interaction with at least one spin within $\Lambda$, and in blue the ones which also share a star interaction.  By definition $\partial\Lambda\cup\Lambda=\emptyset$.
          \label{fig:boundary}}
        \end{figure}
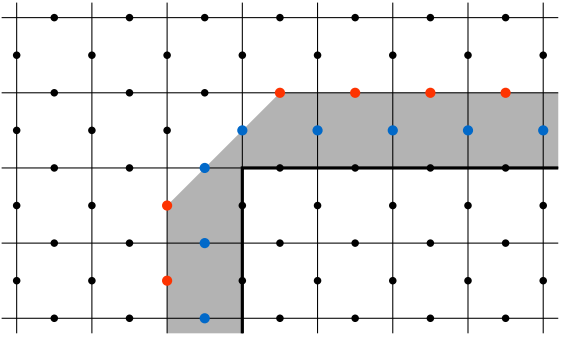

	Given a local Hamiltonian $\op H$ and a finite region $\Lambda$, we define its (exterior) boundary  $\partial \Lambda$  as the set of sites in the complement of $\Lambda$ for which there is an interaction term in $\op H$ acting non-trivially on sites of $\Lambda$ and $\partial \Lambda$ simultaneously, as shown in \cref{fig:boundary}.
	$\bar {\Lambda}$ is defined as $\Lambda \cup \partial \Lambda$ and, for a region $R$,   $\op H_{R}$ will denote the  restriction of $\op H$ to all interactions which are totally contained in $R$.   A \emph{ground state} is then defined as a state functional $\omega$, such that for any finite $\Lambda$, and any local observable $A\in\mathcal B_\Lambda$,
      \begin{equation}\label{eq:gs-condition}
        \omega(A^{\dagger}[\op H_{\bar \Lambda}, A]) \ge 0.
      \end{equation}
      This definition can be obtained by taking the zero temperature limit in the definition of finite temperature equilibrium states as defined by the KMS condition (i.e.\ the limit of increasing-volume Gibbs ensembles satisfying the KMS condition, see \cite[eq.~4.2]{Haag1967}).
      Loosely speaking, it expresses the intuitive understanding that any local perturbation should not decrease the energy of a ground state (see \cite{1608.04449v1}).

      Note that since both $A$ and $\op H_{\bar{\Lambda}}$ have finite support, it is possible to rewrite \cref{eq:gs-condition} in terms of the reduced density matrix of $\omega$ over $\bar \Lambda$, which we denote by $\rho_{\bar \Lambda}$:
      \[
        0\le \omega(A^\dagger[\op H_{\bar \Lambda}, A]) = \tr (\rho_{\bar \Lambda} A^\dagger[\op H_{\bar \Lambda}, A]),
      \]
      or equivalently
      \begin{equation}\label{eq:gs-density-matrix}
        \tr (\rho_{\bar \Lambda} A^\dagger\op H_{\bar \Lambda}A)\ge  \tr (\rho_{\bar \Lambda} A^\dagger A \op H_{\bar \Lambda})
      \end{equation}
      for all $\Lambda$ and all $A\in \bounded_\Lambda$. In turn, this
      implies that
      \begin{equation}\label{eq:gs-channel}
        \tr (\Phi(\rho_{\bar \Lambda}) \op H_{\bar \Lambda})\ge  \tr (\rho_{\bar \Lambda} \op H_{\bar \Lambda}),
      \end{equation}
      for any completely positive, trace preserving linear map $\Phi$,
      as can be seen by applying \cref{eq:gs-density-matrix} to the
      Kraus operators of $\Phi(\cdot) = \sum_i A_i \cdot A_i^\dag$.

      We will argue that the only ground states of the periodic tiling
      Hamiltonian $\op H^{(d)}$ are the ground states of the Toric Code. This
      in turns implies,  that if we take first the limit of
      $N$ going to infinity, and then we send the temperature to zero, we
      recover only ground states of the Toric Code.

      Key to our argument is that part of our Hamiltonian---i.e.\ $\op H_0$---is a ferromagnetic Ising-type interaction, where spin up and down are now the tiling and Toric code subspaces, respectively. We will follow the same proof technique used to show that the 2D Ising model with an external magnetic field has a unique ground state~\cite[ex.~5]{Aizenman1981} to show that any ground state in the thermodynamic limit of our model is completely in the Toric code subspace, by which we mean that for all $\Lambda$, $\omega(\Pi_{\tc,\Lambda}^\perp)=0$ for the projector onto the Toric code subspace $\Pi_{\tc,\Lambda}$ supported on $\Lambda$.

      Let us start with some preliminary observations, which will allow us
      to assume some extra properties of the ground state without loss of generality. Fix $\Lambda$ and let $\{M_i\}_i$ be a decomposition of the identity on $\bar \Lambda$ into orthogonal projectors, such that $\comm{M_i}{\op H_{\bar \Lambda}} = \comm{M_i}{\Pi_{\tc,\Lambda}^\perp} = 0$ for all $i$. We want to show that is sufficient to study $\omega$ restricted to the subspace corresponding to each $M_i$. This is the content of the following lemma.
      \begin{lemma}\label{lemma-signature}
		Let $\omega$ be a ground state. Fix $\Lambda$ and let $\{M_i\}_i$ be as above. Whenever $\omega(M_i) \neq 0$, define $\omega_i$ by
                $$\omega_i(A)=\frac{\omega(M_i A M_i)}{\omega(M_i)}.$$
                Then $\omega_i$  is also a ground state. Moreover, if for every $i$ it holds that $\omega_i(\Pi_{\tc,\Lambda}^\perp)=0$, then also $\omega(\Pi_{\tc,\Lambda}^\perp) = 0$.
      \end{lemma}
      \begin{proof}
        $\omega_i$ is clearly a positive linear functional on local observables so that $\omega_i(\identity)=1$.  It can then be extended to a state on $\mathcal{A}$.
        The fact that $M_i$ commutes with the Hamiltonian makes $\omega_i$ trivially fulfil \cref{eq:gs-condition}, so it is a ground state.
        Finally, we observe that
        \begin{multline*} \omega(\Pi_{\tc,\Lambda}^\perp) = \tr(\rho_{\bar \Lambda} \Pi_{\tc,\Lambda}^\perp ) = \sum_i \tr(M_i \rho_{\bar \Lambda} \Pi_{\tc,\Lambda}^\perp M_i) = \\
        \sum_i \tr(\rho_{\bar \Lambda} M_i \Pi_{\tc,\Lambda}^\perp M_i) = \sum_i \omega(M_i) \omega_i(\Pi_{\tc,\Lambda}^\perp),
        \end{multline*}
        so that the last claim of the lemma follows.
      \end{proof}

      We will use such lemma to make two extra assumptions.
      The first one allows to assume that the ground state is supported, in each site, only in one of the two subspaces (TC or tiling). For that, given a finite region $R\subset \mathbb Z^2$ we consider signatures $\sigma=(\sigma_i)_{i\in R}$ where each $\sigma_i\in \{\text{TC, tiling}\}$. We denote by $P_\sigma$ the projector onto the set of states of signature $\sigma$. It is easy to see that they satisfy the condition of \cref{lemma-signature}. The second assumption is that $\rho_{\bar \Lambda}$ commutes with the Toric Code stabilisers. Again, it is sufficient to consider the projectors onto the eigenspaces of such stabilisers, and the result follows from \cref{lemma-signature}.

      As a second step, we will show that for any ground state for which a square boundary is completely supported in the TC subspace, the interior will be as well; for this we will assume that all square regions have smooth edges as in \cref{fig:boundary}.
      \newcommand{\sigTC}{\mathrm{TC}}
      \newcommand{\sigTiling}{\mathrm{tiling}}

      \begin{lemma}\label{cleaning-lemma}
		Take two concentric square regions $\Lambda' \subsetneq \Lambda$, and a ground state  $\omega$ of $\op H^{(d)}$ with a signature $\sigma$ on $\bar \Lambda$. Assume that $\sigma_s =\sigTC$ on all sites $s$ of $\partial \Lambda'\subset \Lambda\setminus \Lambda'$. Moreover, assume that $\rho_{\bar \Lambda}$ commutes with the Toric Code stabilisers that couple $\Lambda^\prime$ with $\partial \Lambda^\prime$.  Then $\sigma_s=\sigTC$ all sites $s\in \Lambda'$.
      \end{lemma}
      \newcommand{\cptp}{\textsc{CPTP}\xspace}
      \begin{proof}
        Denote with $T\subseteq \Lambda'$ the set of all sites $\sigma\in\Lambda'$ that satisfy $\sigma = \sigTiling$.

        Consider the \cptp map $\Phi_1$ acting on $T$ that on \emph{all} those sites, traces out the tiling sector and replaces it with the maximally mixed state on the TC subspace, i.e.
        \[
          \Phi_1(\rho)=\tr_{T}(\rho) \otimes \qty( \frac{\identity_T^{(\sigTC)}}{\tr \identity_T^{(\sigTC)} }\oplus 0_T^{(\sigTiling)}).
        \]

        Let us now consider a map $\Phi_2$, acting on $\bar \Lambda^\prime$, which implements the following operations: first measures the Toric Code projectors which overlap with $\Lambda^\prime$, and then, conditioned on the syndrome of the measurement, applies a unitary operator which corrects as many as code errors as possible \footnote{One possible way to implement this procedure is to follow a sequence of local steps along an oriented tree structure, as described in~\cite{Dengis2014}}. This can be constructed by choosing as Kraus operators of $\Phi_2$ the product of the projector onto the different syndrome subspaces multiplied on the left with the corresponding unitary operator. We extend this map on the tiling subspace with the identity map, in order to make it a \cptp map.
        Then \cref{eq:gs-channel} implies that
        \begin{equation}\label{eq2:lemma2-thermal}
          \tr (\rho_{\bar \Lambda} \op H_{\bar \Lambda}) \le
          \tr (\Phi_2 \circ \Phi_1 (\rho_{\bar \Lambda}) \op H_{\bar \Lambda}).
        \end{equation}

        We now consider $\tilde \rho_{\bar \Lambda} := \Phi_2 \circ \Phi_1 (\rho_{\bar \Lambda})$.
        For any $\op h$ whose support is disjoint from $\Lambda^\prime$, we have that
       $\tr(\rho_{\bar \Lambda} \op h) = \tr(\tilde \rho_{\bar \Lambda} \op h)$ since $\tilde \rho_{\bar \Lambda^\prime}$ has the same reduced density matrix as $ \rho_{\bar \Lambda^\prime}$ outside of $\Lambda^\prime$.
       Thus \cref{eq2:lemma2-thermal} reduces to $ \tr (\rho_{\bar \Lambda} \op H_{\bar \Lambda^\prime}) \le\tr (\tilde\rho_{\bar \Lambda} \op H_{\bar \Lambda^\prime})$,
        where in $\op H_{\bar \Lambda^\prime}$ only those local Hamiltonian terms appear whose support intersects with $\Lambda^\prime$.
        To finish the proof, we need to find a contradiction assuming that $T$ is not empty.
        First of all, notice that $\tilde \rho_{\bar \Lambda}$ is completely supported on the TC subspace in $\bar \Lambda^\prime$, and that can violate at most 2 of the Toric Code stabilisers (at most one plaquette and one star operator, since any pair of violation would have been destroyed by the action of $\Phi_2$). So $\tr (\tilde\rho_{\bar \Lambda} \op H_{\bar \Lambda^\prime})$ can at most be equal to 2. On the other hand, even with the bonus gained in bottom-left corners of regions supported on the tiling subspace (i.e.\ the bonus of $1/2$ for the all-black tile used to resolve the ground state degeneracy for the periodic tiling pattern), the penalties coming from  mixed signatures in $\tr (\rho_{\bar \Lambda} \op H_{\bar \Lambda'})$ are higher (an overall penalty of at least $7/2$ for each mismatch).
	\end{proof}

	In the next lemma, we generalise the previous one for the case in which some sites on $\partial \Lambda'$ are in the tiling sector.
	\begin{lemma}\label{cleaning-lemma-2}
		Take two concentric square regions $\Lambda' \subsetneq \Lambda$, and a ground state  $\omega$ of $\op H^{(d)}$ with a signature $\sigma$ on $\bar \Lambda$. Moreover, assume that $\rho_{\bar \Lambda}$ commutes with the Toric Code stabilisers that couple $\Lambda^\prime$ with $\partial \Lambda^\prime$. Let $\alpha$ be the number of sites $s\in \partial \Lambda'$ for which $\sigma_s =\sigTiling$, and  $\beta$ the sum of signature mismatches within $\Lambda'$---i.e.\ the number of neighbouring $s,s'\in\Lambda$ for which $\sigma_s\neq\sigma_{s'}$---plus the number of period markers within $\Lambda'$. Then $\beta\le \frac{4}{7}(1+4\alpha)$.
	\end{lemma}
	\begin{proof}
          We follow the same procedure as in the proof for \cref{cleaning-lemma}, obtaining a new state $\tilde \rho_{\bar \Lambda}$ on $\bar \Lambda$,
          such that
          \begin{equation}\label{eq:cleaning-lemma-2-gs}
          \tr(\rho_{\bar \Lambda} \op H_{\bar \Lambda}) \le \tr(\tilde \rho_{\bar \Lambda} \op H_{\bar \Lambda}).
          \end{equation}
          Again, let $T \subset \Lambda^\prime$ the set of sites with tiling signature.
 Let us consider now the interactions $\op h$ in $\op H_{\bar{\Lambda}}$ and compare the values  $ \tr (\tilde{\rho}_{\bar \Lambda}\op h)$ and $ \tr ({\rho}_{\bar \Lambda}\op h)$.  As in previous lemma, if $\op h$ do not overlap with $\Lambda^\prime$, then  $ \tr (\tilde{\rho}_{\bar \Lambda}\op h)= \tr ({\rho}_{\bar \Lambda}\op h)$. Since $\tilde \rho_{\Lambda^\prime}$ is in the TC subspace, and can violate at most 2 stabilisers, its energy can be at most $2+8\alpha$ (the signature in $\partial \Lambda^\prime$ has not changed, and each spin in the tiling subspace can violate up to 4 Ising-type interactions).
 On the other hand, since there are at least $\beta$ signature mismatches for $\rho_{\Lambda^\prime}$, and each of them has an energy of at least $7/2$ (again, this is lower than $4$ because of the $1/2$ bonus given to the all-black tile), we have that
$\tr(\rho_{\bar \Lambda}\op H_{\bar{\Lambda^\prime}}) \ge \frac{7}{2} \beta $. Inserting these two bounds into \cref{eq:cleaning-lemma-2-gs}, we obtain the desired bound.
        \end{proof}

      In the following, we will show that if we pick the outer square in \cref{cleaning-lemma} large enough, we can always find an inner concentric square---of at least a third of the outer square's size---for which we can then apply \cref{cleaning-lemma} or \cref{cleaning-lemma-2}. In the pictures of the following lemma, we have coloured with black the spins which are in the TC subspace, and in yellow the ones which are not.

      \begin{lemma}\label{lemma:ring-cut}
		Take some square $\Lambda$ of side length $3k$, where $k=10^6 N_d^2$ and consider a ground state $\omega$ of $\op H^{(d)}$ with a signature $\sigma$ on $\bar \Lambda$.
        Subdivide $\Lambda$ into $k\times k$ squares:
        \[
          \tikz[draw=white]{
            \node[anchor=south west,inner sep=0] at (-1/6, -1/6) {\includegraphics[width=3.33333cm]{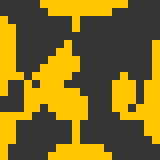}};
            \draw (0, 0) rectangle (3, 3);
            \draw (1, 1) rectangle (2, 2);
            \draw[dashed] (0, 1) -- (3, 1);
            \draw[dashed] (0, 2) -- (3, 2);
            \draw[dashed] (1, 0) -- (1, 3);
            \draw[dashed] (2, 0) -- (2, 3);
            \draw[decorate, decoration={brace, amplitude=1mm}, black] (-.3, 0) -- (-.3, 1) node[midway, xshift=-8] {$k$};
            \draw[black] (1.6, 1.6) -- (2.5, 2.2) -- (4, 2.2) node[anchor=west] {\upshape{centre} $C$};
          }
        \]
        Then $\sigma_i=\text{TC}$ for all $i$ in the centre $C$.
      \end{lemma}
      \begin{proof}
            We start with a few preliminary observations, and we refer the reader to \cref{fig:boundary} for verification.  The boundary $\partial\Lambda$ contains precisely $3\times 3k$ spins on each side, and thus $12\times 3k=36k$ spins overall.  We denote this spin count with $|\partial\Lambda|$ in the context of this proof.
            By \cref{cleaning-lemma-2} (assuming that for all $s\in\partial\Lambda:\sigma_s=\sigTiling$), we thus know that we can have \emph{at most} $\frac{4}{7}(1+4\times |\partial\Lambda|)<83k$ penalties from signature mismatches or period markers within $\Lambda$.

    		The overall area of $\Lambda$ encompasses $9k^2$ tiles, and we count $|\Lambda|=2\times 9k^2 + 2\times 3k$ spins. Every sub-square of size $N_d\times N_d$ which is not fully in the TC subspace carries a penalty $\ge1$---note that this holds true regardless of the bonus terms present in the tiling, as the period penalty and TC-tiling mismatch are both larger. This means that at most a fraction of
    		$$
    		\frac{83k}{|\Lambda|/N_d^2}=\frac{83 N_d^2}{18k+6}=\frac{83}{18\times10^6 + 6/N_d^2}<\frac{1}{10\,000}
    		$$
    		of \emph{spins} $s\in\Lambda$ can have signature $\sigma_s=\sigTiling$.  For a subset $A\subset\Lambda$, we denote this fraction with $f(A)$.

		So let us assume that $f(\Lambda)<1/10\,000$. Take $\Lambda$ and shrink it uniformly by at most $k/10$, by which we mean we shrink the square on each side by one tile at a time, i.e.\ while keeping smooth boundaries as in \cref{fig:toriccode}.  This sweeps a region $A$ which covers at least $1/10$ of the area of $\Lambda$, and thus in particular the number of spins $|A|\ge|\Lambda|/10$.  Since $f(\Lambda)<1/10\,000$, it follows that $f(A)<1/1000$.  This immediately implies that there exists a square $\Lambda'$ concentric with $\Lambda$ and such that $A\subset\Lambda'\subset \Lambda$, and which satisfies $f(\partial\Lambda')<1/1000$.

        \DeclareDocumentCommand{\crossFig}{ m m O{red} O{.067} }{%
          \draw[#3, thick] (#1-#4, #2-#4) -- (#1+#4, #2+#4);
          \draw[#3, thick] (#1-#4, #2+#4) -- (#1+#4, #2-#4);
        }
        \[
          \raisebox{-.45\height}{\tikz[draw=white]{
              \node[anchor=south west,inner sep=0] at (-1/6, -1/6) {\includegraphics[width=3.33333cm]{figures/gs-fin.png}};
              \draw (0, 0) rectangle (3, 3);
              \draw (1, 1) rectangle (2, 2);
              \draw[dashed] (0, 1) -- (3, 1);
              \draw[dashed] (0, 2) -- (3, 2);
              \draw[dashed] (1, 0) -- (1, 3);
              \draw[dashed] (2, 0) -- (2, 3);
              \foreach \x/\y in {1/0, 2/0, 3/.33333, 3/1.5, 0/.16667, 0/.83333, 0/1.5, 0/2.66667, .83333/3, 2.16667/3}
              \crossFig{\x}{\y};
              \draw[red] (1, 0) -- (2, 0);
              \draw[red] (3, .33333) -- (3, 1.5);
              \draw[red] (0, .16667) -- (0, .83333);
              \draw[red] (0, 1.5) -- (0, 2.66667);
              \draw[red] (.83333, 3) -- (2.16667, 3);
              \node[white] at (1.75, 1.75) {$C$};
              \node[white] at (2.75, 2.75) {$\Lambda$};
          }}
          \quad\mapsto\quad
          \raisebox{-.45\height}{\tikz[draw=white]{
              \node[anchor=south west,inner sep=0] at (-1/6, -1/6) {\includegraphics[width=3.33333cm]{figures/gs-fin.png}};
              \draw (1/6, 1/6) rectangle (3-1/6, 3-1/6);
              \draw (1, 1) rectangle (2, 2);
              \draw[dashed] (0, 1) -- (3, 1);
              \draw[dashed] (0, 2) -- (3, 2);
              \draw[dashed] (1, 0) -- (1, 3);
              \draw[dashed] (2, 0) -- (2, 3);
              \foreach \x/\y in {1.33333/.16667, 1.5/.16667, 2.83333/.66666, 2.83333/1, .16667/.66666, .16667/.83333, .16667/1.5, .16667/1.66666, 1.33333/2.83333, 1.5/2.83333}
              \crossFig{\x}{\y};
              \draw[red] (1.33333, .16667) -- (1.5, .16667);
              \draw[red] (2.83333, .66666) -- (2.83333, 1);
              \draw[red] (.16667, .66666) -- (.16667, .83333);
              \draw[red] (.16667, 1.5) -- (.16667, 1.66666);
              \draw[red] (1.3333, 2.83333) -- (1.5, 2.83333);
              \node[white] at (1.75, 1.75) {$C$};
              \node[white] at (2.6, 2.6) {$\Lambda'$};
          }}
        \]

        Two things may happen. If in between the centre square $C$ and $\Lambda'$ there exists another square $\Lambda''$ (i.e.\ with $\partial\Lambda''\subset\Lambda'\setminus C$) such that for any $s\in\partial \Lambda''$, we have $\sigma_s=\sigTC$, then  \cref{cleaning-lemma} immediately implies that $\sigma_i=\sigTC$ for all $i$ in the centre $C$, and the claim follows.

		It remains to analyse the case where no such square $\Lambda''$ exists.  We first apply \cref{cleaning-lemma-2} again, this time to $\Lambda'$: as $f(\partial\Lambda')<1/1000$, and $|\partial\Lambda'|\le|\partial\Lambda|\le 36k$, we know that there exist \emph{at most} $36k/1000\times 8\le k/3$ spins $s\in\Lambda'$ with $\sigma_s=\sigTiling$.  But since no $\Lambda''$ exists with a boundary $\partial\Lambda''$ completely with tiling signature, there have to be \emph{at least} $k-k/10=9k/10$ spins within $\Lambda'$ with a tiling signature (one within the boundary for each concentric square between $\Lambda'$ and $C$).  Contradiction, and the claim follows.
      \end{proof}

      All the above results lead trivially to
      \begin{corollary}
        All ground states of $\op H^{(d)}$ are fully supported in the TC subspace.
      \end{corollary}

\section{Extensive Hamiltonian Perturbation}

We will prove that even under extensive Hamiltonian perturbation the ground state of our model
remains in the same phase as the Toric Code. In order to achieve this, we consider local perturbations of the form
$$\op V = \sum_{u,r} \op V_{u,r},$$
where $\op V_{u,r}$ acts on sites at a distance at most
$r$ from $u$, and $\norm{V_{u,r}}\le J e^{-\mu r}$ for some positive $J$ and $\mu$.
We
want to show that for the local Hamiltonians $\op H$ we constructed, in the
thermodynamic limit and for
sufficiently small $J$, the ground states of $\op H + \op V$ are automorphically equivalent to the ground states of $\op H$ (in the sense
of \cite{Bachmann_2011}), i.e.\ that the ground states can be connected by locally-generated unitary transformations.

We will show this by using a modified version of the proof presented in \cite{Bravyi_2011}, from
where we take the notation and terminology in what follows, but accounting for the following key differences:
\begin{enumerate}
	\item we have open and not periodic boundary conditions,
	\item  our Hamiltonian is not a sum of projectors, and
	\item before the threshold system size, the ground state does
	not have topological order.
\end{enumerate}
In particular, regarding 1., a solution of the stability problem for topological models besides periodic
boundaries is given in \cite{Stability_unpublished}. For the convenience of the reader, instead of just referring to \cite{Stability_unpublished},
we detail here the concrete solution for our particular case of interest.

Let us define, for every finite and rectangular $\Lambda$ with side $L$ larger than the threshold $N_d$, and $s\in[0,1]$,
$$\op H_\Lambda(s) = \op H^{(d)}_\Lambda + s\!\!\!\!\!\!\!\!\! \sum_{u:d(u,\Lambda^c) > d_L}\!\!\!\!\!\!\!\!\! \op V_u,$$
i.e.\ we only include in $\op H(s)$ the perturbation
terms which are at distance $d_L$ (to be determined later) away from the boundary.
We will show that, if $J$ is smaller than some $J_0$ (independent of $\Lambda$), then
$\op H(s)$ will have a spectral gap of $\order{1}$ for every $s\in [0,1]$, and this
is sufficient to prove the equivalence of the ground states of $\op H^{(d)} $ and $\op H^{(d)} +
\op V$ in the thermodynamic limit \cite{Bachmann_2011}. For the rest of the proof, we will fix
$\Lambda$ and write $\op H_s$ instead of $\op H_\Lambda(s)$ for simplicity.

The proof strategy of \cite{Bravyi_2011} can be summarised as follows: the aim is to show that, if
$J$ is sufficiently small, and if $\op H_s$ has a spectral gap of at least $1/2$ for every
$s \in [0,s^*]$, then the spectral gap at $s^*$ is actually larger than $3/4$. This ``bootstrapping'' procedure then
immediately implies that the gap of $\op H_s$ is never zero for every $s\in [0,1]$. Our Hamiltonian
$\op H_0$ does not have a spectral gap of $1$ as is the case in the original proof, but since it
is still a constant (independent of $\Lambda$) a simple rescaling will be sufficient to follow the
proof strategy.

As a first step to prove the ``bootstrapping'' argument one uses the so called quasi-adiabatic
evolution \cite{Bravyi_2010}, which allows to write a general perturbation as a block-diagonal
operator with respect to the projector $P_0$ on the ground state of $\op H_0$, plus some weak
boundary terms; then one uses the Topological Quantum Order of the Toric Code in order to decompose a
block-diagonal perturbation into a locally block-diagonal part plus a small residual perturbation.

To simplify notation, for each site $u\in \Lambda$ we will denote $d_u = d(u,\Lambda^c)$ the distance of
$u$ from the boundary of $\Lambda$. Given an Hermitian operator $O$, we will often use the
decomposition $O = \sum_{u,r} O_{u,r}$, where $u$ runs over the vertices in $\Lambda$ and $O_{u,r}$
is supported on a ball of radius $r$ around $u$. If $\norm{O_{u,r}}\le J f(r)$ for some positive $J$
and some function $f(r)$ with values in $[0,1]$ and decaying faster than any polynomial function, we will say that $O$
\emph{is quasi-local with strength $J$}. The operator $O_u = \sum_{r} O_{u,r}$, in which every term of the decomposition
acts on $u$, will be said to be supported \emph{around $u$}.

\begin{lemma}
  \label{lemma:tqo-stab-1}
  For every $s\in[0,1]$, $\op H_s $ has the same spectrum as
  \[ \op H_s^\prime = \op H_0 + \!\!\!\!\!\!\! \sum_{u: d_u > d_L} \!\!\!\!\!\!\! X_u(s) ,\]
  where $\comm{X_u(s)}{P_0} = 0$ and each $X_u(s)$ is a quasi-local operator with strength
  $\order{J}$ supported around $u$.
\end{lemma}

\begin{proof}
  We will write for simplicity
  $\op H_0 = \sum_{u} h_u$, where each $h_u$ acts around $u$ (so that each plaquette and star
  operator is associated to one site belonging to them).
  The proof will follow \cite[Lemma 7]{Bravyi_2011}, with the necessary modifications:
  we consider $U_s$ the quasi-adiabatic evolution
  (also known as spectral flow) associated to the Hamiltonian path $\op H_s$, generated by the
  operator
  \[ i  D_s =  \int_{-\infty}^\infty \dd t F(t) e^{i \op H_s t} \op V e^{-i \op H_s t}, \]
  where $F(t)$ is the weight function defined in \cite{Bravyi_2011}.
  The unitary $U_s$ is then defined by
  \[ U_s = \mcl S \exp( i \int_0^s \dd s^\prime D_{s^\prime} ), \]
  where $\mcl S$ denotes the $s$-ordered exponential. It is the unique solution to the differential
  equation
  \begin{align*}
    \dv{s} U_s &= i D_s U_s,\\
    U_0 &= \identity.
  \end{align*}
  Since it is locally generated, it satisfies Lieb-Robinson bounds \cite{Lieb_1972, Hastings_2006, Bravyi_2006,
    Nachtergaele_2006}, meaning that for $s\in[0,1]$ a quasi-local operator with strength $J$ supported around $u$ will be
  mapped by $U_s$ to another $\order{J}$-strength quasi-local perturbation supported around $u$, thus preserving locality.
  Moreover, we see that each $D_s(u) = -i \int \dd t F(t) e^{i \op H_s t} \op V_u e^{-i \op H_s t} $ is a quasi-local
  $\order{J}$-strength operator supported around $u$.
  We can then define
  \[ \op H_s^\prime = U_s^\dag \op H_s U_s.  \]
  Clearly $\op H_s^\prime$ has the same spectrum as $\op H_s$, while its ground state projector is
  $P_0$. We will use frequently the fact that the ground state projector $P_0(s)$ of $\op H_s$
  satisfies $U_s^\dag P_0(s) U_s = P_0$.
  Let us define
  \[ \mcl F_s(O) = \int_{-\infty}^\infty \dd{t}g(t) e^{i \op H_s t} O e^{-i \op H_s t}, \] where $g(t)$ is
  the filter function defined in \cite[Lemma 7]{Bravyi_2011}.
  It has the property that $\comm{\mcl F_s(O)}{P_0(s)} = 0$,
  and therefore $\comm{U_s^\dag \mcl F_s(O) U_s}{P_0} = 0$.
  Since $\op H_s$ has exponentially decaying interactions, $\mcl F_s$ will map quasi-local operators of strength $\order{1}$ supported around $u$ to
  quasi-local operators of strength $\order{J}$ supported around $u$ \cite[Lemma 2]{Bravyi_2011}. Moreover, it is
  easy to see that
  $\op H_s^\prime = U_s^\dag \mcl F_s(U_s \op H_s^\prime U_s^\dag) U_s$, so we have that
  \[ \op H_s^\prime = \op H_0 + \sum_{u \in \Lambda} \tilde h_u(s) + \sum_{u: d_u > d_L}\tilde V_u(s) ,\]
  where
  \begin{align*}
    \tilde h_u(s) & =  \qty[ U_s^\dag \mcl F_s(h_u) U_s - h_u]; \\
    \tilde{V}_u(s) & =  s U_s^\dag \mcl F_s(\op V_u) U_s.
  \end{align*}
  Each of the terms appearing in the above decomposition commutes with $P_0$.
  For each $v\in \Lambda$ such that $d_v \le d_L$, let $\pi(v)$ be the closest point to $v$ in $\Lambda$ such
  that $d_{\pi(v)} > d_L$ (and we make an arbitrary choice if it is not unique). Let $N(u) =
  \pi^{-1}(u)$ be the set of points which are sent to $u$ by $\pi$.
  We can then define for each $u$ such that $d_u > d_L$:
  \[ X_u(s) = \tilde h_u(s) + \tilde V_u(s) + \sum_{v \in N(u)} \tilde h_v(s).\] In order to
  conclude, we have to show that $\tilde h_u(s)$ and $\tilde V_u(s)$ have strength $\order{J}$ if
  $d_u > d_L$, while $\tilde h_v(s)$ is supported around $u=\pi(v)$ with strength
  $\order{Jf(d(u,v))}$ if $d_v \le d_L$, for some fast decaying function $f(d)$.
  This will guarantee that the sum $\sum_{v \in N(u)} \tilde h_v(s)$ is a quasi-local operator of
  strength $\order{J}$ around $u$.

  In order to do so, let us observe that if a
  quasi-local operator $O_u = \sum_{r} O_{u,r}$ such that $\norm{O_{u,r}} \le J f(r)$ satisfies that
  $O_{u,r} = 0$ for each $r \le r_0$, then it is also a quasi-local operator with strength $Jf(r_0)$
  supported at any point at distance $r_0$ from $u$.

  The terms $\tilde V_u(s)$ are supported around $u$ with $d_u > d_L$ by construction, and
  they are quasi-local with strength $\order{J}$ because of Lieb-Robinson bounds applied to
  $\mcl F_s$ and to $U_s$. Let us now analyse instead the terms of $\tilde h_u(s)$.

  For each $u$, we will further decompose
  $ \tilde h_u(s) = h^1_u(s) + h^2_u(s)$ as follows
  \begin{align*}
    h^1_u(s) = & U_s^\dag \mcl F_s(h_u - U_s h_u U_s^\dag) U_s ;\\
    h^2_u(s) = &  U^\dag_s\qty[\mcl F_s(U_s h_u U_s^\dag) - U_s h_u U_s^\dag] U_s.
  \end{align*}
  We will now threat the two terms independently. Let us start by observing that
  \[ h_u - U_s h_u U_s^\dag = i \int_0^s\dd{ s^\prime}U_{s^\prime} \comm{D_{s^\prime}}{h_u} U^\dag_{s^\prime}. \]
  Because of the quasi-locality of $D_{s^\prime}$, if $d_u > d_L$ the commutator
  $\comm{D_{s^\prime}}{h_u}$ will be a quasi-local Hamiltonian of strength $\order{J}$ supported around
  $u$. Applying again Lieb-Robinson bounds for $U_s$, this also hold for $h^1_u(s)$.

  If instead $d_u \le d_L$, we can expand $\comm{D_{s^\prime}}{h_u}$ in terms of the local
  decomposition of $D_{s^\prime} = \sum_{v,r} D_{v,r}(s^\prime)$. Any term whose support is
  disjoint from the support of $h_u$ will not contribute to the commutator, and in particular only
  terms with $r\ge d_L-d_u$ will be present. By the previous observation, it is supported around $\pi(u)$.
  Applying a the unitary rotation $U_s$ will preserve the locality and the strength, so also
  $h^1_u(s)$ will be of strength $\order{Jf(d_L-d_u)}$ supported around $\pi(u)$.

  Similarly we observe that
  \begin{align*}
    \mcl F_s(O) & = \int_{-\infty}^{\infty}\dd{t}g(t) e^{i \op H_s t} O e^{-i \op H_s t}\\
                            & = O + i \int_{-\infty}^{\infty} \dd{t}g(t) \int_0^t \dd \tau e^{i
                              \op H_s \tau } \comm{\op V}{O}e^{-i \op H_s \tau} \\
                            & = O + \int_{-\infty}^{\infty} \dd{\tau}\tilde g(\tau) e^{i \op H_s \tau}
                              \comm{\op V}{O}e^{-i \op H_s \tau},
  \end{align*}
  with
  \[ \tilde g(\tau) = \begin{cases}
      i \int_\tau^\infty\dd{t}g(t) & \text{if } \tau \ge 0,\\
      -i \int_{-\infty}^\tau\dd{t}g(t) & \text{if } \tau < 0,
    \end{cases}
  \]
  which decays faster than any polynomial. Therefore
  \begin{multline*}
    \mcl F_s(U_s h_u U_s^\dag) - U_s h_u U_s^\dag \\
    = \int_{-\infty}^\infty\dd{t}\tilde g(t) e^{i \op H_s t} U_s \comm{U_s^\dag \op V U_s}{h_u}
    U_s^\dag e^{-i \op H_s t}.
  \end{multline*}
  As we did before, we now use the fact that $U_s^\dag \op V U_s$ has a quasi-local structure, so that we
  can analyse two different cases: if $d_u > d_L$, then $\comm{U_s^\dag \op V U_s}{h_u}$ will be
  quasi-local with strength $\order{J}$, while if $d_u \le d_L$ there will be no terms with support
  smaller than $d_L -d_u$, so it is supported around $\pi(u)$ with strength $\order{Jf(d_L-d_u)}$.
  Again, applying the rotation $U_s^\dag \cdot U_s$ will not change these
  properties, so we have also proven that they hold for $h^2_u(s)$. This concludes the proof.
\end{proof}

For the second part of the proof, we will need the following lemma instead of \cite[Lemma 3]{Bravyi_2011}.

\begin{lemma}
  \label{lemma:tqo-stab-2}
  Fix a site $u$ and let $X_u = \sum_{r}X(r)$ be a quasi-local operator with strength $J$ supported
  around $u$, and satisfying $\comm{X_u}{P_0} = 0$.
  Then we can decompose $X_u$ as
  \[ X_u = \sum_{r\ge N_d} W_{u,r} + \Delta_u,  \]
  where $\comm{W_{u,r}}{P_0} = 0$, and $W_u = \sum_{r} W_{u,r}$ has strength $\order{J}$, and
  $\norm{\Delta_u}$ decays faster than any power of $d_u$.
\end{lemma}
\begin{proof}
  We decompose $X_u = \sum_{r} X(r)$, and by assumption $\norm{X(r)}$ decays faster than any power
  of $r$. Since we choose the side of $\Lambda$ to be larger than $N_d$, the ground state $P_0$ of
  $\op H_0$ is the Toric Code with open boundary condition, which satisfies topological quantum
  order for operators that do not act on the boundary spins. Therefore, for each $r < d_u$,
  we have that $P_0X(r)P_0$ is a constant multiple of $P_0$. By adding constants we can assume
  $P_0 X(r)P_0=0$. We define $\Delta_u = P_0X_u = \sum_{r \ge d_u}P_0X(r)P_0 $ and $X_u^\prime = X_u - \Delta_u$.
  By construction, we have that $P_0 X_u^\prime = 0$ and that $\norm{\Delta_u}$ decays faster than any
  power of $d_u$. We want to show now that we can treat $X_u^\prime$ in the same way as in \cite[Lemma 3]{Bravyi_2011},
  obtaining a decomposition $X_u^\prime = \sum_{u, r}W_{u,r}$, with $\comm{W_{u,r}}{P_0}=0$ and of
  strength $\order{J}$.
  The original proof requires a property, denoted TQO-2, which our model satisfies only partially:
  if $A$ is a square contained in $\Lambda$, and $C$ is the square containing the first and second
  neighbours of $A$ in $\Lambda$, and we denote by $P_C$ the projector on the ground state
  of the Hamiltonian $\op H_0$ restricted to $C$ (i.e.\ considering only the interaction terms which
  intersect $A$ and are contained into $C$), then TQO-2 implies that for every operator $O_A$
  supported on $A$ it holds that \cite[Corollary 1]{Bravyi_2011}
  \begin{equation}\label{eq:tqo-2}
    \norm{O_A P_0} = \norm{O_A P_C}.
  \end{equation}
  This is only true for our model when $C$ has side larger than $N_d$, since in that case $P_C$ is
  the projector on the Toric Code ground state in $C$, which does satisfy TQO-2.
  We can therefore apply \cref{eq:tqo-2} to the operator
  $\sum_{r=1}^l X_u^\prime(r)$ if $N_d \le l < d_u$, so in that case we obtain the
  following bound
  \begin{multline} \norm{\sum_{r=1}^l X_u^\prime(r) P_{B_u(l+2)}} = \norm{\sum_{r=1}^l
      X_u^\prime(r) P_0} \\
      = \norm{X_u^\prime P_0} + \norm{\sum_{r\ge l} X_u^\prime P_0} \le J f(l), \label{eq:tqo-2-cor}
    \end{multline}
    for some $f(l)$ decaying faster than any polynomial, and $B_u(l)$ is a ball of radius $l$
    centred at $u$.
    We will now show how to decompose $X_u^\prime$.
    We can consider the following decomposition of the identity $\identity = \sum_{m=N_d}^{N+1}
    E_m$, where $N$ is such $B_u(N) = \Lambda$,
    \begin{align*}
      E_{N_d} &= \identity - P_{B_u(N_d)}, & \\
      E_{m} &= P_{B_u(m-1)} - P_{B_u(m)}, & \text{for } N_d < m \le N,\\
      E_{N+1} &= P_{B_u(N)} = P_0. &\\
    \end{align*}
    Since $E_{N+1} X^\prime_u = X^\prime_u E_{N+1} = 0$, we have that
    \[ X_u^\prime = \sum_{\substack{1 \le q \le N\\ N_d\le p,r \le N }} E_p X_u^\prime(q) E_r =
      \sum_{j \ge 2N_d} Y(j) + \sum_{j \ge N_d} Z(j),
    \]
    where
    \begin{align*}
      Y(j) = \sum_{\substack{N_d \le p,r \le N \\ p+r=j } } E_p \sum_{q=1}^{\max (p,r) - 2} X_u^\prime(q) E_r, \\
      Z(j) = \sum_{N_d \le p,r, \le j+1} E_p X_u^\prime(j) E_r .
    \end{align*}
    Both $Y(j)$ and $Z(j)$ are Hermitian, they annihilate $P_0$, and they are supported on
    ${B_{j+1}(u)}$. We are only left to show that their norm is decaying fast in $j$.
    This is trivial for $Z(j)$, since its norm is bounded by the norm of $X_u^\prime(j)$, while each
    of the $j-1$ terms in $Y(j)$ corresponding to the different choices of $p$ and $r$ can be
    bounded using \cref{eq:tqo-2-cor}. Defining $W_{u,r} = Y(r-1) + Z(r-1)$ concludes the proof.
\end{proof}

Finally, we want to apply \cite[Lemma 5]{Bravyi_2011} in order to prove that a perturbation
$\sum_{u}\sum_{r\ge N_d} W_{u,r}$ which does not intersect the boundary and such that
$\comm{P_0}{W_{u,r}} = 0$ is relatively bounded by $\op H_0$---we can use standard perturbation
theory for the terms which do intersect the boundary. Again, we run into the problem that our model
does not satisfy TQO-2 for regions smaller than $N_d$. But we can trivially modify the proof of such
Lemma, and in particular of \cite[Proposition 1]{Bravyi_2011}, by requiring that $r\ge N_d$ (since by
construction we only have operators $W_{u,r}$ with such restriction on $r$). In that case, given a
partition of $\Lambda$ into disjoint rectangular boxes $B_1 \cup \dots \cup B_M$ of size $r\ge N_d$, we can
define for every binary string $\mcl Y \in \{0,1\}^M$ an operator
\[ R_{\mcl Y} = \prod_{a=1}^M [ \mcl Y_a (1-P_{B_a}) + (1-\mcl Y_a) P_{B_a} ].\]
We say that a box $B_a$ is occupied if $\mcl Y_a =1$.
We then see that if $W$ is an operator acting on a square $A$ of size $r$, which does not intersect
the boundary of $\Lambda$, and such that $WP_0 = P_0W = 0$, then $R_{\mcl Y} W
R_{\mcl Z} \neq 0$ if and only if $A$ intersects a box occupied in $\mcl Y$, a box occupied in $\mcl
Z$, and the only differences in the configurations $\mcl Y$ and $\mcl Z$ are in boxes intersecting $A$.
From here we can follow verbatim the proof of \cite[Proposition 1]{Bravyi_2011}.
We can conclude then with the following lemma.

\begin{lemma}
  There exists a positive function $G(L)$, growing slower than any polynomial in $L$, such that
  if for $\Lambda$ of size $L$ we have that $d_L \ge G(L)$ then the spectrum of $\op H_s$ is contained in
  \[ \bigcup_k [k(1-b)-\delta, k(1+b)+\delta]\]
  for every $s\in [0,1]$, where $k$ runs over the spectrum of $\op H_0$, $b=\order{J}$, and $\delta$
  vanishes as $L$ goes to infinity.
\end{lemma}
\begin{proof}
  Using \cref{lemma:tqo-stab-1} and \cref{lemma:tqo-stab-2}, we see that the spectrum of $\op H_s$
  is the same as the spectrum of
  \[ \op H_0 + \sum_{u: d_u > d_L} (W_u + \Delta_u), \]
  where $W_u = \sum_r W_{u,r}$ has strength $\order{J}$ and satisfies $\comm{W_{u,r}}{P_0}=0$,
  while $\norm{\Delta_u}$  decays faster than any power of $d_u$.
  Then by \cite[Lemma 5]{Bravyi_2011} with the modification presented before, the spectrum of $\op H_0 + \sum_u \sum_{r \le d_u} W_{u,r}$ is
  contained in $\bigcup_k [k(1-b), k(1+b)]$ for some $b=\order{J}$.

  The rest of the perturbation $R = \sum_{u: d_u > d_L} \qty(\Delta_u + \sum_{r > d_u} W_{u,r})$
  is treated by bounding its norm. Since in $\op H_s$ we only included
  perturbations that originally act at distance at least $d_L$ from the boundary, and that at any
  given distance $d$ from the boundary of $\Lambda$ there are $\order{L-d}$ sites, we have that the
  total norm of the perturbation $R$ can be bounded by $C\sum_{d\ge d_L} (L-d) f(d)$,
  where $f(d)$ decays faster than any polynomial and $C$ is some positive constant. Therefore, being
  the tail of a discrete convolution, $\sum_{d\ge d_L} (L-d) f(d)$ decays faster than any polynomial
  in $d_L$, while grows polynomially in $L$. We can then choose $d_L$ growing slower than any power
  and still guarantee that there exists a constant $\delta$ which vanishes in the limit
  $L \to \infty$, such that the spectrum of $\op H_s$ is contained in
  $\bigcup_k [ k(1-b)-\delta, k(1+b) + \delta]$.
\end{proof}

\end{document}